\documentclass[a4paper,11pt]{article}

\usepackage{amsmath,amsthm,graphics,color}
\usepackage{graphicx}

\newtheorem{theorem}{Theorem}
\newtheorem{observation}{Observation}
\newtheorem{proposition}{Proposition}

\newtheorem{lemma}{Lemma}

\newcommand{\NP}{{\sf NP}}
\newcommand{\FPT}{{\sf FPT}}

\newcommand{\dist}{{\rm dist}}

\begin{document}

\title{Parameterized Algorithms for Finding 
Square~Roots\footnote{
The research leading to these results has received funding from the European Research Council under the European Union's Seventh Framework Programme (FP/2007-2013)/ERC Grant Agreement n. 267959. The research also has been supported by EPSRC (EP/G043434/1) and ANR Blanc AGAPE (ANR-09-BLAN-0159-03).
A preliminary version of this paper appeared as an extended abstract in the proceedings
of WG 2013 \cite{CochefertCGKP13}.}
}

\author{
Manfred Cochefert\thanks{Laboratoire d'Informatique Th\'eorique et Appliqu\'ee, Universit\'e de Lorraine, 57045 Metz Cedex 01, France, 
\texttt{\{manfred.cochefert, jean-francois.couturier, dieter.kratsch\}@univ-lorraine.fr}}
\addtocounter{footnote}{-1}
\and
Jean-Fran\c{c}ois Couturier\footnotemark
\and
Petr A. Golovach\thanks{Department of Informatics, University of Bergen, PB 7803, 5020 Bergen, Norway,
\texttt{petr.golovach@ii.uib.no}}
\addtocounter{footnote}{-2}
\and
Dieter Kratsch\footnotemark
\and
\addtocounter{footnote}{1}
Dani{\"e}l Paulusma\thanks{School of Engineering and  Computing Sciences, Durham University, Science Laboratories, South Road, Durham DH1 3LE, UK,
\texttt{daniel.paulusma@durham.ac.uk}}
}

\date{}

\maketitle

\begin{abstract}
We show that the following two problems are fixed-parameter tractable with parameter $k$:
testing whether a connected $n$-vertex graph with $m$ edges has a square root with at most $n-1+k$ edges and testing
whether such a graph has a square root with at least $m-k$ edges. Our first result implies that squares of graphs obtained from trees by adding at most $k$ edges
can be recognized in polynomial time for every fixed $k\geq 0$; previously this result was known only for $k=0$.
Our second result is equivalent to stating that deciding whether a graph can be modified into a square root of itself by at most $k$ edge deletions is fixed-parameter tractable with parameter~$k$.
\end{abstract}

\section{Introduction}
\label{sec:intro}
Squares and square roots are classical concepts in graph theory that are defined as follows. The {\it square} $G^2$ of a
graph $G=(V_G,E_G)$ is the graph with vertex set $V_G$ such that any two distinct vertices 
$u,v\in V_G$ are adjacent in $G^2$ if and only if $u$ and $v$ are of distance at most 2 in $G$.
A graph $H$ is a {\it square root} of $G$ if $G=H^2$. There exist graphs with no square root,
graphs with a unique square root as well as graphs with many square roots.

Mukhopadhyay~\cite{Mukhopadhyay67} showed in 1967 that  
a connected graph $G$ with $n$ vertices $v_1,\ldots,v_n$  has a square root if and only if there exists a set of $n$ complete subgraphs $K^1,\ldots,K^n$ of $G$ with $\bigcup_iV_{K^i}=V_G$ such that 
$K^i$ contains $v_i$ for all $1\leq i\leq n$, and  $K^i$ contains $v_j$ if and only if $K^j$ contains $v_i$ for all $1\leq i<j\leq n$.
This characterization did not yield a polynomial time algorithm for recognizing squares. 
In fact, in 1994, Motwani and Sudan~\cite{MotwaniS94} showed that the {\sc Square Root} problem, which is that of testing whether a graph has a square root,
is \NP-complete. 
 This fundamental result triggered a
lot of research on the computational complexity of recognizing squares of graphs and computing square
roots under the presence of additional structural assumptions. In particular, the following   
two recognition questions have attracted attention; here ${\cal G}$ denotes some fixed graph class.  
\begin{itemize}
\item[(1)] How hard is it to recognize squares of graphs of~$\cal G$?
\item[(2)] How hard is is to recognize  graphs of $\cal G$ that have a square root?
\end{itemize} 
Ross and Harary~\cite{RossH60} characterized squares of a tree and
proved that if a connected graph has a unique tree square root, then this root is unique up to isomorphism.  
Lin and  Skiena~\cite{LinS95} gave linear time algorithms for recognizing squares of trees and planar graphs with a square root.
The results for trees~\cite{LinS95,RossH60} were generalized to block graphs by Le and Tuy~\cite{LeT10}.
Lau~\cite{Lau06} gave a polynomial time algorithm for recognizing squares of bipartite graphs.
Lau and Corneil~\cite{LauC04} gave a polynomial time algorithm for 
recognizing squares of proper interval graphs and showed that the problems of recognizing squares of chordal graphs,  squares of split graphs, and chordal
graphs with a square root are all three \NP-complete.
Le and Tuy~\cite{LeT11} gave a quadratic time algorithm for recognizing squares of strongly chordal split graphs.   
Milanic and Schaudt~\cite{MilanicS13} gave linear time algorithms for recognizing trivially perfect graphs and threshold graphs with a square root.
Adamaszek and Adamaszek ~\cite{AdamaszekA11} proved that if a graph has a square root of  girth at least 6, then this square root is unique up to isomorphism. Farzad, Lau, Le and Tuy~\cite{FarzadLLT12} showed that recognizing graphs with a square root of girth at least $g$ is polynomial-time solvable if $g\geq 6$ and \NP-complete if $g=4$. The missing case~$g=5$ was shown
to be \NP-complete by Farzad and Karimi~\cite{FarzadK12}.

\subsection{Our Results}\label{s-our}
The classical {\sc Square Root} problem is a decision problem. We introduce two optimization variants of it in order to be able to take a {\it parameterized} road to square roots. A problem with  input size $n$ and a parameter $k$
is said to be \emph{fixed parameter tractable} (or \FPT) if it can be solved in time $f(k)\cdot
n^{O(1)}$ for some function $f$ that only depends on $k$. 
We consider two natural choices for the parameter $k$ for our optimization variants of the {\sc Square Root} problem and in this way obtain the first FPT algorithms for 
square root problems.

First, in Section~\ref{s-min}, we parameterize the {\sc Minimum Square Root} problem, which is that of testing whether a graph has a square root
with at most $s$ edges for some given integer $s$. 
Because any square root of a connected $n$-vertex graph $G$ is a connected spanning subgraph of $G$, every square root of  $G$ has at least $n-1$  edges. 
Consequently, any instance $(G,s)$ of {\sc Minimum Square Root} with $s\leq n-2$ is a no-instance
if $G$ is connected, 
which means that we may assume that $s\geq n-1$.
Hence, $k=s-(n-1)$ is the natural choice of parameter.
Our main result is that  {\sc Minimum Square Root} is \FPT\ with parameter $k$ 
\footnote{We restrict ourselves to connected graphs for simplicity. We may do this for the following reason. For disconnected $n$-vertex graphs with $\ell\geq 2$ connected components the natural parameter is $k=s-(n-\ell)$ instead of $k=s-(n-1)$. Because a square root of a graph is the disjoint union of square roots of its connected components, our \FPT\ result for connected graphs immediately carries over to disconnected graphs if we choose as parameter $k=s-(n-\ell)$ instead.}.
We prove this result by showing that an instance  of {\sc Minimum Square Root} can be reduced to an instance of a more general problem, in which
we impose additional requirements on some of the edges, namely to be included or excluded from the square root.
We prove that the new instance has size quadratic in $k$.
In other words, we show that {\sc Minimum Square Root} has a generalized kernel of quadratic size (see Section~\ref{sec:defs} for the definition of this notion). 
This result is further motivated by the observation
that {\sc Minimum Square Root} generalizes the
 problem of recognizing squares of trees (take $s=n-1$).
A weaker statement of our \FPT\ result is that of saying that the problem of recognizing squares of graphs
of the class $${\cal G}_k=\{G\; |\; G\; \mbox{is a graph obtainable from a tree by adding at most $k$ edges}\}$$ is polynomial-time solvable for all fixed $k\geq 0$.
As such, our result can also be seen as an extension of the aforementioned result of recognizing squares of trees~\cite{LinS95}.

Second, in Section~\ref{s-max}, we parameterize the {\sc Maximum Square Root} problem, which is that of testing whether a given graph $G$ with $m$ edges has a square root
with at least $s$ edges for some given integer $s$. We show that this problem is \FPT\ with parameter $k=m-s$. This choice of parameter is also natural, as $G$ has a square root
with at least $s$ edges if and only if $G$ can be modified into a square root (of itself) by at most $k$ edge deletions. Hence, our second \FPT\ result can be added to the growing body of parameterized results for graph editing problems, which form a well studied problem area within algorithmic graph theory. 
In Section~\ref{s-max} we also present an exact exponential time algorithm for  {\sc Maximum Square Root}, which could be seen
as an improvement of the algorithm implied by the characterization of Mukhopadhyay~\cite{Mukhopadhyay67}. 

In Section~\ref{s-con} we mention a number of relevant open problems.

\subsection{Preliminaries}\label{sec:defs}
We only consider finite undirected graphs without loops and multiple edges. 
We refer to the textbook by Diestel~\cite{Diestel10} for any undefined graph terminology and
to the textbooks of Downey and Fellows~\cite{DowneyF99},  Flum and Grohe~\cite{flum-grohe-book}, and   Niedermeier~\cite{niedermeier-book} for  detailed introductions  to parameterized complexity theory.

Let $G$ be a graph. We denote the vertex set and edge set of $G$ by $V_G$ and $E_G$, respectively. The subgraph of $G$
induced by a subset $U\subseteq V_G$ is denoted by $G[U]$. 
The graph $G-U$ is the graph obtained from $G$ by removing all vertices in $U$. If $U=\{u\}$, we also write $G-u$.
The \emph{distance} $\dist_G(u,v)$ between a pair of vertices $u$ and $v$ of
$G$ is the number of edges of a shortest path between them. 
The \emph{open neighborhood} of a vertex $u\in V_G$ is defined as $N_G(u) = \{v\; |\; uv\in E_G\}$, and
its \emph{closed neighborhood} is defined as $N_G[u] = N_G(u) \cup \{u\}$. 
Two vertices $u,v$ are said to be \emph{true twins} if $N_G[u]=N_G[v]$, and $u,v$ are \emph{false twins} if $N_G(u)=N_G(v)$.
A vertex $u$ is \emph{simplicial}, if $N_G(u)$ is a clique.
The {\it degree} of a vertex
$u\in V_G$ is denoted $d_G(u)=|N_G(u)|$.
The maximum degree of $G$ is denoted $\Delta(G)=\max\{d_G(v)|v\in V_G\}$.
A vertex of degree 1 is said to be a \emph{pendant} vertex. 

Let $G$  be a connected graph. Let $S\subset V_G$, and let $X$ and $Y$ be two disjoint
nonempty vertex subsets of  $G-S$.
Then $S$ is a {\it separator} of $G$ if $G-S$ is disconnected, 
$S$ is an \emph{(X,Y)-separator} if $G-S$ has no path that connects a vertex of $X$ to a vertex of  $Y$, and 
$S$ is a \emph{minimal} $(X,Y)$-separator  if $S$ is an $(X,Y)$-separator of $G$ and no proper subset of $S$ is an  $(X,Y)$-separator.  
Moreover, $G$ is {\it $2$-connected} if and only if $|V_G|\ge 3$ and $G$ has no separators 
of size one.

The {\it union} of two graphs $G_1$ and $G_2$ is the graph $(V_{G_1}\cup V_{G_2},E_{G_1}\cup E_{G_2})$.
The graph $K_n$ denotes the complete graph on $n$ vertices. The graph $K_{1,r}$ denotes the star on $r+1$ vertices.

A well-known technique to show that a parameterized problem $\Pi$ is 
fixed-parameter tractable is to find 
a \emph{reduction to a problem kernel}. This technique replaces an instance
$(I, k)$ of $\Pi$ with a reduced instance $(I', k')$ of $\Pi$ called a
\emph{(problem) kernel} such that the following three conditions hold:
\begin{itemize}
\item[i)] $k'\leq k$ and $|I'|\leq g(k)$
for some computable function $g$;
\item[ii)] the reduction from $(I, k)$ to
$(I', k')$ is computable in polynomial time;
\item[iii)] $(I,k)$ is a {\tt yes}-instance of $\Pi$ if and
only if $(I',k')$ is a {\tt yes}-instance of $\Pi$.
\end{itemize} 
If we slightly modify this definition by letting the instance $(I',k')$ belong to a different problem than~$\Pi$, then $(I',k')$ is called a {\it generalized} kernel for $\Pi$ in the literature. This concept has been introduced and named {\it bikernel} by Alon, Gutin, Kim, Szeider and Yeo~\cite{AlonGKSY11}; 
a related notion is compression.  
An upper bound $g(k)$ on $|I'|$ is called the \emph{kernel size}, and a kernel is called \emph{linear}  if its size is linear in~$k$ and \emph{quadratic} if its
size is quadratic in~$k$. 
It is well known that a parameterized problem
is fixed-parameter tractable if and only if it has a kernel (see for example~\cite{niedermeier-book}).

\section{The Minimum Square Root Problem}\label{s-min}

As discussed in Section~\ref{s-our}, we consider connected graphs only and parameterize {\sc Minimum Square Root} by $k=s-(n-1)$.   
From now on we denote this problem as
 
\medskip
\noindent
{\sc Trees $+\;k$ Edges Square Root}\\
{\it Input:} \hspace*{8.5mm}a connected
graph $G$ and an integer $k\geq 0$\\
{\it Parameter:} $k$\\
{\it Question:} \hspace*{1.5mm} has $G$ a square root with at most $n-1+k$ edges?

\medskip
\noindent
We show the following result.

\begin{theorem}\label{thm:tree-few-edges} 
The {\sc Tree $+\;k$ Edges Square Root} problem can be solved in time $2^{O(k^4)} + O(n^4m)$ on graphs with $n$ vertices and $m$ edges.
\end{theorem}

The remainder of this section is organized as follows. In Section~\ref{s-structural} we show a number of structural results needed to prove Theorem~\ref{thm:tree-few-edges}. In Section~\ref{s-reduction} we consider the more general problem

\medskip
\noindent
{\sc Tree $+\;k$ Edges Square Root with Labels}\\
{\it Input:} \hspace*{8mm} a connected
graph $G$, an integer $k\geq 0$ and 
two disjoint
subsets
\\ \hspace*{2.0cm}
$R,B\subseteq E_G$\\
{\it Parameter:} $k$.\\
{\it Question:} \hspace*{1.5mm} has $G$ a square root $H$ with at most $n-k+1$ edges, such that 
\hspace*{2.0cm}$R\subseteq E_H$ and $B\cap E_H=\emptyset$?

\medskip
\noindent
Note  that the sets $R$ and $B$ in this problem are given sets of {\it required} edges (that have to be in the square root) and
{\it blocked} edges (that are not allowed to be in the square root), respectively. Also note that 
{\sc Tree $+\;k$ Edges Square Root with Labels}
generalizes {\sc Trees $+\;k$ Edge Square Root}; choose $R=B=\emptyset$.
We reduce  {\sc Tree $+\;k$ Edges Square Root} to {\sc Tree $+\;k$ Edges Square Root with Labels}  where the size of the graph in the obtained instance is $O(k^2)$. 
In other words, we construct a quadratic generalized kernel for {\sc Tree $+\;k$ Edges Square Root}. 
This means that to solve an instance of {\sc Trees $+\;k$ Edge Square Root}, 
we can solve the obtained instance of  {\sc Tree $+\;k$ Edges Square Root with Labels} by a brute
force algorithm. In Section~\ref{s-solving} we analyze the corresponding running time and complete the proof of Theorem~\ref{thm:tree-few-edges}.

\subsection{Structural Results}\label{s-structural}

We start with the following observation that we will frequently use.

\begin{observation}\label{obs:leaves}
Let $H$ be a square root of a connected graph $G$. 
\begin{itemize}
\item[i)] If $u$ is a pendant vertex of $H$, then $u$ is a simplicial vertex of $G$. 
\item[ii)] If $u,v$ are pendant vertices of $H$ adjacent to the same vertex, then $u,v$ are true twins in $G$.
\item[iii)] If $u,v$ are pendant vertices of $H$ adjacent to different vertices, then $u$ and $v$ are not adjacent in $G$ unless $H=K_2$.
\end{itemize}
\end{observation}

We now state five useful lemmas, the first two of which, Lemmas~\ref{lem:leaves-one} and~\ref{lem:leaves-two}, can be found implicitly in the paper of Ross and Harary~\cite{RossH60}.
Ross and Harary~\cite{RossH60} consider tree square roots, whereas we are concerned with finding general square roots. As such we give explicit statements of Lemmas~\ref{lem:leaves-one} and~\ref{lem:leaves-two}. We also give a proof of Lemma~\ref{lem:leaves-two} (the proof of Lemma~\ref{lem:leaves-one} is straightforward).

\begin{lemma}\label{lem:leaves-one}
Let $H$ be a square root of  a graph $G$.
Let $\{u_1,\ldots,u_r\}\subseteq V_H$ for some $r\geq 3$ induce a star in $H$ with central vertex $u_1$.
Let $u_3,\ldots,u_r$ be pendant and $\{u_2\}$ be a 
$(\{u_1,u_3,\ldots,u_r\}, 
V_H\setminus\{u_1,\ldots,u_r\})$-separator
of $H$. 
Then  $\{u_1,\ldots,u_r\}$ is a clique of $G$, and $\{u_1,u_2\}$ is a minimal 
$(\{u_3,\ldots,u_r\},
 V_G\setminus\{u_1\ldots,u_r\})$-separator
of $G$. 
\end{lemma}

\begin{lemma}\label{lem:leaves-two}
Let $G$ be a connected graph with a square root $H$.
Let $\{u_1,\ldots,u_r\}$, $r\geq 3$ be a clique in $G$, such that $\{u_1,u_2\}$ is a minimal 
$(\{u_3,\ldots,u_r\}, V_G\setminus\{u_1,\ldots,u_r\})$-separator of $G$. 
Let  $\{x_1,\ldots,x_p\}=N_G(u_1)\setminus\{u_1,\ldots,u_r\}$ 
for some $p\geq 1$
and $\{y_1,\ldots,y_q\}=N_G(u_2)\setminus\{u_1,\ldots,u_r\}$ 
for some $q\geq 1$, 
as shown in Figure~\ref{fig:trimming}.
Then the following three statements hold:
\begin{itemize}
\item[i)]   $u_1u_2\in E_H$ and,  
either $u_3u_1,..,u_ru_1\in E_H$,  $u_3u_2,\ldots,u_ru_2\notin E_H$, $u_1x_1,\ldots, $ $u_1x_p\notin E_H$,  and $\{u_2\}$ is a minimal $(\{u_1,u_3,\ldots,u_r\}, V_H\setminus\{u_1,\ldots,u_r\})$-separator in $H$,
 or 
 $u_3u_1,\ldots,u_ru_1\notin E_H$, $u_3u_2,\ldots,u_ru_2\in E_H$, $u_2y_1,\ldots,u_2y_q\notin E_H$
 and $\{u_1\}$ is a minimal $(\{u_2,..,u_r\}, V_H\setminus\{u_1,..,u_r\})$-separator in $H$
(see Figure~\ref{fig:trimming-i-ii-iii}~i)).
\item[ii)] If $u_1,u_2$ are true twins in $G$, then either $u_1x_1,\ldots,u_1x_p\in E_H$ or $u_2x_1,\ldots,u_2x_p\in E_H$.
Moreover, in this case, $G$ is the union of two complete graphs with vertex sets $\{u_1,\ldots,u_r\}$ 
 and $\{u_1,u_2,x_1,\ldots,x_p\}$, respectively, and $G$ 
 has two (isomorphic) square roots with edge sets $\{u_1u_2,\ldots,u_1u_r\}$ $\cup\{u_2x_1,\ldots,u_2x_p\}$ and 
 $\{u_2u_1,u_2u_3,\ldots,u_2u_r\}$ $\cup\{u_1x_1,\ldots,u_1x_p\}$, respectively 
(see Figure~\ref{fig:trimming-i-ii-iii}~ii)).
\item[iii)] If $N_G[u_2]\setminus N_G[u_1]\neq\emptyset$, then  $u_2u_1,\ldots,u_ru_1\in E_H$,  $u_3u_2,\ldots,u_ru_2\notin E_H$, $u_1x_1,\ldots,u_1x_p\notin E_H$. Moreover, the graph $H'$ obtained from $H$ by deleting all $u_iu_j$ with $3\leq i<j\leq r$
is a square root of $G$ (in which $\{u_1,\ldots,u_r\}$ induces a star with central vertex $u_1$ and with leaves $u_2,u_3,\ldots,u_r$ that 
are pendant vertices except for $u_2$ 
(see Figure~\ref{fig:trimming-i-ii-iii}~iii)). 
\end{itemize}
\end{lemma}

\begin{figure}[h]
\centering\scalebox{1}{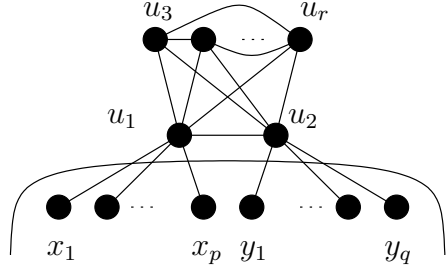}
\caption{The graph $G$ of Lemma~\ref{lem:leaves-two}. Note that $p\geq 1$ and $q\geq 1$, because
$\{u_1,u_2\}$ is a minimal $(\{u_3,\ldots,u_r\}, V_G\setminus\{u_1,\ldots,u_r\})$-separator of $G$.
Also note that  $x_i=y_j$ for some $1\leq i\leq p$ and $1\leq j\leq q$ is possible.
\label{fig:trimming}}
\end{figure}

\begin{figure}[!h]
\centering\scalebox{1}{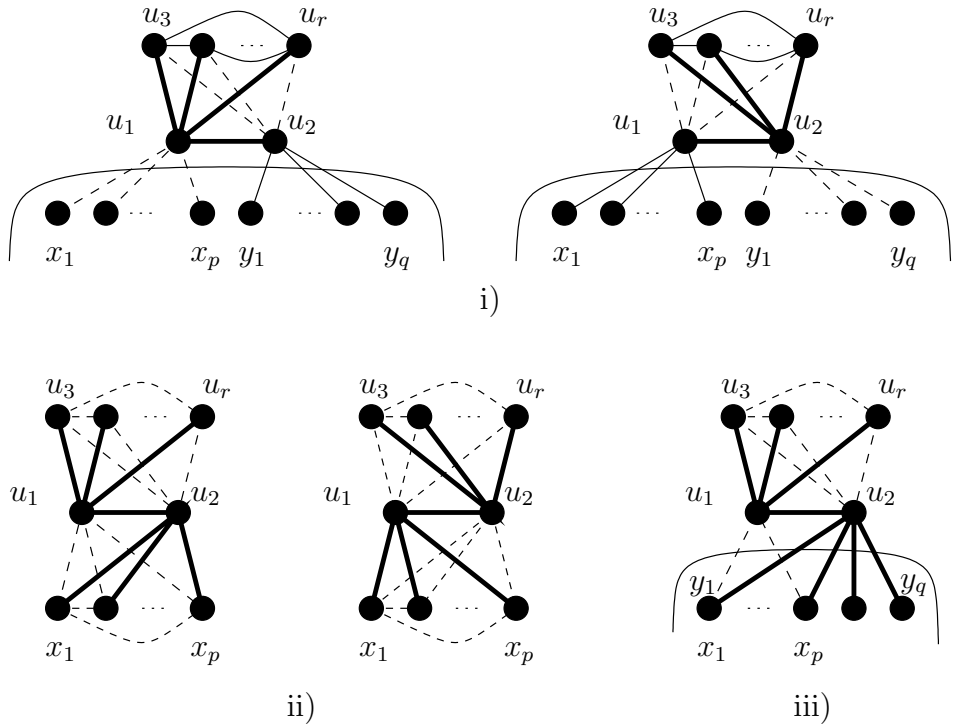}
\caption{Square roots of $G$ corresponding to statements i)--iii) of Lemma~\ref{lem:leaves-two}, respectively. Edges of $G$ that belong to the square roots are shown by thick lines, whereas edges of~$G$ that do not belong to the square roots are shown by dashed lines. In i), edges of $G$ that may be put in a square root of~$G$ are shown by thin lines. The square roots in~ii) and~iii) are the 
specific square roots defined in statements~ii) and~iii) of Lemma~\ref{lem:leaves-two}, respectively.
\label{fig:trimming-i-ii-iii}}
\end{figure}

\begin{proof}
We first prove i).
As $\{u_1,u_2\}$ is a 
$(\{u_3,\ldots,u_r\}, V_G\setminus\{u_1,\ldots,u_r\})$-separator of $G$, at least one vertex $u_i$ with $r\geq 3$ is adjacent to one of $u_1,u_2$ in $H$, say to $u_1$.
Then $u_1x_1,\ldots,u_1x_p\notin E_H$; otherwise, that is, if $u_1$ is adjacent to some $x_j$ in $H$, then 
$u_ix_j\in E_G$ contradicting the fact that $\{u_1,u_2\}$ is a $(\{u_3,\ldots,u_r\}, V_G\setminus\{u_1,\ldots,u_r\})$-separator of $G$.
Because  $u_1x_1,\ldots,u_1x_p\notin E_H$, at least one vertex $y_h$ must be adjacent to $u_2$ in $H$ (as otherwise $H$ is not connected and hence cannot 
be the square root of $G$, which is  a connected graph).
Because $\{u_1,u_2\}$ is  a $(\{u_3,\ldots,u_r\}, V_G\setminus\{u_1,\ldots,u_r\})$-separator of $G$, this means that
$u_3u_2,\ldots,u_ru_2\notin E_H$.
Consequently, $u_1u_2\in E_H$ and $\{u_2\}$ is a minimal $(\{u_1,u_3,\ldots,u_r\}, V_H\setminus\{u_1,\ldots,u_r\})$-separator in $H$.
Suppose that there is a vertex $u_i$, $3\leq i\leq r$, such that $u_iu_1\notin E_H$. Since $u_3,\ldots,u_r$ are not adjacent to $u_2$, it follows that any $(u_2,u_i)$-path in $H$ has length at least 3, which is not possible as $u_2u_i\in E_G$. 
We conclude that $u_3u_1,\ldots,u_ru_1\in E_H$. Hence we have shown i).

We now prove ii). Note that $\{x_1,\ldots,x_p\}=\{y_1,\ldots, y_q\}$ with $p=q$.  
Due to i) either $u_1$ or $u_2$ is not adjacent to any $x_i$. In the first case $u_2$ must be adjacent to all $x_i$ in $H$, as otherwise there is no required path of length at most $2$ in $H$ between some $x_i$ and $u_1$. Similarly, in the second case, $u_1$ must be adjacent to all $x_i$ in $H$.
Hence, $\{u_1,u_2,x_1,\ldots,x_p\}$ is a clique in $G$. 
If $H$ has an edge $x_iz$ with $z\notin \{u_1,\ldots,u_r,x_1,\ldots,x_p\}$, then $zu_2\in E_G$, which is not possible.
This means that $G$ is the union of two complete graphs with vertex sets $\{u_1,\ldots,u_r\}$ 
 and $\{u_1,u_2,x_1,\ldots,x_p\}$, respectively. It is readily seen that
 $G$ has two (isomorphic) square roots with edge sets $\{u_1u_2,\ldots,u_1u_r\}$ $\cup\{u_2x_1,\ldots,u_2x_p\}$ and 
 $\{u_2u_1,u_2u_3,\ldots,u_2u_r\}$ $\cup\{u_1x_1,\ldots,u_1x_p\}$, respectively. Hence we have shown ii).

It remains to prove~iii). Let $y_i\in N_G[u_2]\setminus N_G[u_1]\neq\emptyset$. Due to i) we have that
$u_1u_2\in E_H$, and that either $u_3u_1,..,u_ru_1\in E_H$,  $u_3u_2,\ldots,u_ru_2\notin E_H$, $u_1x_1,\ldots, $ $u_1x_p\notin E_H$, or
 $u_3u_1,\ldots,u_ru_1\notin E_H$, $u_3u_2,\ldots,u_ru_2\in E_H$, $u_2y_1,\ldots,u_2y_q\notin E_H$.
 If the latter case holds, then any $(u_2,y_i)$-path in $H$ has length at least 3, which is not possible as $u_2y_i\in E_G$. Hence the former case must hold.  
Let $H'$ be a graph obtained from $H$ by deleting all $u_iu_j$ for $i,j\in\{3,\ldots,r\}$. It is readily seen that $H'^2=H^2=G$.
Hence we have shown~iii).
\end{proof}

Let $G$ be a graph that contains (besides possibly some other vertices) $p+q+r$ distinct vertices $u_1,\ldots,u_r$, $x_1,\ldots,x_p$, $y_1,\ldots y_q$ for some $r\geq 3$, $p\geq 1$ and $q\geq 1$, such that the following conditions hold:
\begin{itemize}
\item[i)]
 $\{u_1,\ldots,u_r\}$ is a clique in $G$;
 \item[ii)]  $\{u_1,u_2,u_3\}$ is a minimal $(\{u_4,\ldots,u_r\},V_G\setminus\{u_1,\ldots,u_r\})$-separator in~$G$ if 
$r\geq 4$;
\item[iii)] 
$\{u_1,u_3,\ldots,u_r\}\cup \{x_1,\ldots,x_p\}\cup \{y_1,\ldots,y_q\}=N_G(u_2)$;
\item[iv)] $\{u_2,u_4,u_5,\ldots,u_r\}=N_G(u_1)\cap N_G(u_3)$;
\item[v)] 
$\{x_1,\ldots,x_p\}\subseteq N_G(u_1)$ and $\{y_1,\ldots,y_q\}\subseteq N_G(u_3)$;
\item[vi)]
$x_iy_j\notin E_G$ for $i=1,\ldots, p$ and $y=1,\ldots,q$.
\end{itemize}
We call $G$ an {\it $F$-graph} and $\{u_1,u_2,u_3\}$ an {\it $F$-triple} with {\it outer vertices} $u_1$ and $u_3$, see Figure~\ref{f-fgraph} for an example. Here,  $F$ refers to the graph in Figure~\ref{fig:path}. These notions are further explained by Lemmas~\ref{lem:path-one} and~\ref{lem:path-two}.

\begin{figure}[ht]
\centering\scalebox{0.8}{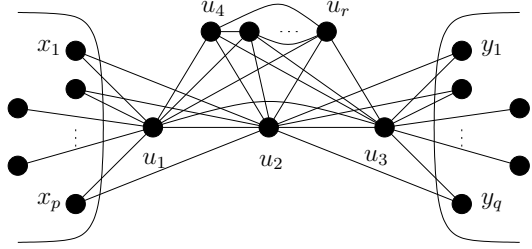}
\caption{An example of an $F$-graph with $r\geq 4$. Note that there are no edges between any two vertices $x_i$ and $y_j$. Also note that
the two outer vertices $u_1$ and $u_3$ of the $F$-triple $\{u_1,u_2,u_3\}$ 
may be adjacent to vertices not adjacent to $u_2$ (but they may not have any common neighbor in
$\{x_1,\ldots,x_p\}\cup \{y_1,\ldots,y_q\}$).
\label{f-fgraph}}
\end{figure}

\begin{figure}[ht]
\centering\scalebox{0.8}{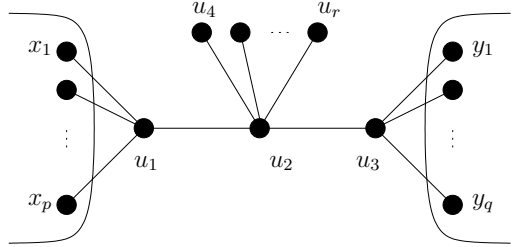}
\caption{The graph $F=F(p,q,r)$ with $p\geq 1$, $q\geq 1$ and $r\geq 3$; if $r=3$ then $F$ does not contain any pendant vertices $u_4,\ldots,u_r$.
Here, we depicted $F$ as a subgraph of the graph $H$ in Lemma~\ref{lem:path-one}. To be more precise, the graph $F$ is exactly the graph with the black edges. In $H$ the vertices $u_1,\ldots,u_r$ have only neighbors that are in $F$, whereas a vertex $x_i$ or $y_j$ may have one or more neighbors in $H$ that are outside $F$; however, no $x_i$ and $y_j$ have a common neighbor in $H$. Moreover, in $H$, the only edges incident to vertices in $F$ are the black edges depicted (edges of $F$) and possibly some edges between two vertices $x_i,x_j$ or between two vertices $y_i$,$y_j$; such edges have not been depicted in the figure.
\label{fig:path}}
\end{figure}

\begin{lemma}\label{lem:path-one}
Let $H$ be a square root of a graph $G$. 
Let $H$ contain the graph~$F$ of Figure~\ref{fig:path} as a subgraph, such that $u_4,\ldots,u_r$ are pendant vertices of~$H$ (if $r\geq 4$), $d_H(u_2)=r-1$,
$u_1u_2u_3$ is an induced path in $H$ that is not contained in any cycle of length at most $6$, 
$\{x_1,\ldots,x_p\}=N_H(u_1)\setminus\{u_2\}$ and  $\{y_1,\ldots,y_q\}=N_H(u_3)\setminus\{u_2\}$.
Then $G$ is an $F$-graph.
\end{lemma}

\begin{proof}
Conditions (i)-(iii)~and~(v) are readily seen to hold. Conditions~iv) and~vi)  follow from the condition that the path $u_1u_2u_3$ is not contained in any cycle of length at most 6 in $H$.
\end{proof}

\begin{lemma}\label{lem:path-two}
Let $G$ be a connected $F$-graph. If $H$ is a square root of $G$, then the graph $F$ of Figure~\ref{fig:path} is a subgraph of $H$ such that $d_H(u_2)=r-1$, 
$\{x_1,\ldots,x_p\}=N_H(u_1)\setminus\{u_2\}$ and  $\{y_1,\ldots,y_q\}=N_H(u_3)\setminus\{u_2\}$.
Moreover, the graph  obtained from $H$ by deleting all edges $u_iu_j$ with $4\leq i<j\leq r$
is~a square root of $G$ 
that contains $u_4,\ldots,u_r$ as pendant vertices (if $r\geq 4$).
\end{lemma}

\begin{proof}
Let $H$ be a  square root of $G$.
 We consider the following three cases.

\medskip
\noindent{\bf Case 1.} $u_1u_2,u_2u_3\in E_H$. 
Because $x_1u_3,\ldots,x_pu_3\notin E_G$, this means that $x_1u_2,\ldots,x_pu_2\notin E_H$. 
Symmetrically,  $y_1u_2,\ldots,y_qu_2\notin E_H$. 
Since each 
$x_iu_2\in E_G$ but $x_iu_2\notin E_H$, $H$ has an $(x_i,u_2)$-path of length 2. 
Because $d_G(u_2)=p+q+r-1$, the middle vertex of this path is in $\{u_1,u_3,\ldots,u_r\}$.
Because $x_i$ is not adjacent to $u_3,\ldots,u_r$ in $H$ (as it not so in $G$), this path goes through $u_1$. 
In other words, $x_1u_1,\ldots,x_pu_1\in E_H$ and, by symmetry, $y_1u_3,\ldots,y_qu_3\in E_H$.
If a vertex $z\notin \{u_2,x_1,\ldots,x_p\}$ is adjacent to $u_1$ in $H$, then $z$ is adjacent to both $u_2$ and $x_1$
in $G$. Because 
$d_G(u_2)=p+q+r-1$, we find that $z\in \{u_3,\ldots,u_r\}$ or $z\in \{y_1,\ldots,y_q\}$. However, none of $\{u_3,\ldots,u_r\}$ is adjacent to $x_1$,
whereas none of $\{y_1,\ldots,y_q\}$ is adjacent to $u_2$. We conclude that
$\{x_1,\ldots,x_p\}=N_H(u_1)\setminus\{u_2\}$ and by using the same arguments that $\{y_1,\ldots,y_q\}=N_H(u_3)\setminus\{u_2\}$.

Now we show that $u_4u_2,\ldots,u_ru_2\in E_H$. To prove it, assume that some $u_i$, $4\leq i\leq r$, is not adjacent to $u_2$ in $H$. Then $u_1$ and $u_i$ are at distance at least 3 in $H$ contradicting $u_1u_i\in E_G$.
We already deduced that $x_1u_2,\ldots,x_pu_2\notin E_H$ and that $y_1u_2,\ldots,y_qu_2\notin E_H$.
By assumption, $u_2$ is adjacent to both $u_1$ and $u_3$. As $d_G(u_2)=p+q+r-1$, we then find that $d_H(u_2)=r-1$.

To conclude the proof for this case, it remains to observe that if some $u_i,u_j$ are adjacent in $H$ for $i,j\in\{4,\ldots,r\}$, then the graph $H'$ obtained from $H$ by the removal of these edges is a square root of $G$. 

\medskip
\noindent{\bf Case 2.} $u_1u_2,u_2u_3\notin E_H$. 
Since $u_1u_2\notin E_H$, $u_1u_2\in E_G$ and $d_G(u_2)=p+q+r-1$, there exists a vertex
$z\in \{x_1,\ldots,x_p\}\cup\{u_4,\ldots,u_r\}$ such that
$u_1z,zu_2\in E_H$. 
Because $z$ is not adjacent to $y_1,\ldots,y_q$ in $G$, we find that $y_1u_2,\ldots,y_qu_2\notin E_H$. 
By the same arguments, we obtain  $x_1u_2,\ldots,x_pu_2\notin E_H$.
Hence, $z\in\{u_4,\ldots,u_r\}$. By symmetry, some vertex from $\{u_4,\ldots,u_r\}$ is adjacent to $u_3$ in $H$.
Consequently, each vertex of $\{u_1,u_2,u_3\}$ is adjacent to some vertex in $\{u_4,\ldots,u_r\}$ in $H$. As $\{u_1,u_2,u_3\}$ separates
$\{u_4,\ldots,u_r\}$ from $V_G\setminus \{u_1,\ldots,u_r\}$, this means that $H$ has no edges that join $u_1,u_2,u_3$ with the vertices of $V_G\setminus \{u_1,\ldots,u_r\}$;  a contradiction. Hence, this case is not possible.

\medskip
By symmetry, it remains to consider the following case.

\medskip
\noindent{\bf Case 3.} $u_1u_2\in E_H$ and $u_2u_3\notin E_H$. Because $u_1u_2\in E_H$ and 
$y_1u_1,\ldots,y_qu_1\notin E_G$, we find that $y_1u_2,\ldots,y_qu_2\notin E_H$. 
Because $y_1u_2\in E_G$, this means that $H$ contains a $(y_1,u_2)$-path of length 2. 
Because $u_2u_3\notin E_H$ and $d_G(u_2)=p+q+r-1$, such a path should go through one of the 
vertices of $\{u_1,u_4,\ldots,u_r\}\cup\{x_1,\ldots,x_p\}$.
However, none of these vertices is adjacent to $y_1$ in $G$, and consequently not in $H$ either; a contradiction. Therefore, this case is not possible either.
\end{proof} 

\begin{lemma}\label{lem:twins}
Let $u,v$ be true twins in a connected graph $G$ with at least three vertices. Let $G'$ be the graph obtained from $G$ by deleting $v$. 
The following two statements hold:
\begin{itemize}
\item [i)] If $H'$ is a square root of $G'$, then the graph $H$ obtained from $H'$ by adding $v$ with $N_{H}(v)=N_{H'}(u)$ (that is, by adding a false twin of $u$) is a square root of $G$.
\item [ii)] If $H$ is a square root of $G$ such that $u,v$ are false twins in $H$, then the graph $H'$ obtained by deleting $v$  is a square root of $G'$.
\end{itemize}
\end{lemma}

\begin{proof}
We first prove i).
Let $H'$ be a square root of $G'$, and let $H$ be the graph obtained from $H'$ by adding a false twin $v$ of $u$. As $G$ is a connected graph with at least three vertices, $u$ is adjacent in $H'$ to some vertex $z$. Then $u$ and $v$ are adjacent to $z$ in $H$ and thus $d_H(u,v)\le 2$. Hence, $uv$ is an edge of $H^2$.
Then it is straightforward to see that $G=H^2$. 
Statement~ii) follows from the fact that identifying false twins does not change the distance between any two vertices.
\end{proof}

\subsection{Construction of the Generalized Kernel}\label{s-reduction}
As discussed, in this section, we reduce {\sc Tree $+\;k$ Edges Square Root} to {\sc Tree $+\;k$ Edges Square Root with Labels}  in such a
way that the size of the graph in the obtained instance is 
$O(k^2)$.

First, we informally sketch the main steps of the reduction. 
Let $G$ be a connected graph with $n$ vertices, and let $k$ be a positive integer.

\begin{figure}[ht]
\centering\scalebox{0.75}{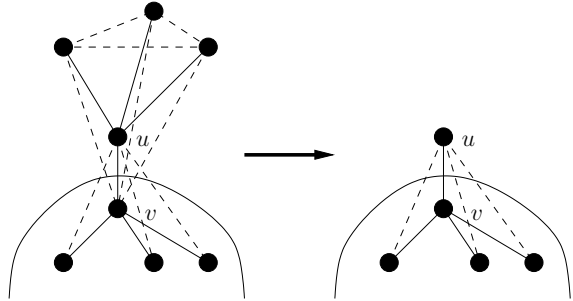}
\caption{Trimming; the edges of $H$ are shown by the solid lines.
\label{fig:trim-informal}}
\end{figure}

Suppose that $H$ is a square root of $G$ with at most $n+k-1$ edges. If $H$ has a vertex $u$ of degree at least 2 that has exactly one non-pendant neighbor~$v$, then we recognize the corresponding structure in $G$ and delete those vertices of $G$ that are pendant vertices of $H$ adjacent to $u$ as shown in Figure~\ref{fig:trim-informal}, that is,
similar to the algorithm of Lin and Skiena~\cite{LinS95}, we ``trim'' pendant edges in potential roots.
Since the root we are looking for is not a tree, our trimming is more sophisticated and
based on Lemmas~\ref{lem:leaves-one} and \ref{lem:leaves-two}. We will show that in this way we obtain a graph $G'$ with
$n'$ vertices that has the following  
property: every pendant vertex of any square root $H'$ of $G'$ with at most $n'-1+k$ edges is adjacent to a vertex that has at least two non-pendant neighbors in~$H'$. 

Suppose that $H'$ has a sufficiently long induced path $P$, such that every internal vertex of $P$ has exactly two non-pendant neighbors in $H'$. Let $u$ be an internal vertex of $P$, and let $x,y\in V_P$ be the two non-pendant neighbors of $u$. Using
Lemmas~\ref{lem:path-one} and~\ref{lem:path-two}, we recognize the corresponding structure in $G'$ and modify $G'$ as shown in 
Figure~\ref{fig:path-informal}, that is, we delete $u$ in $H'$ and join $x$ any $y$ by an edge. By performing this operation recursively, we obtain a graph $G''$ with $n''$ vertices. 

\begin{figure}[ht]
\centering\scalebox{0.75}{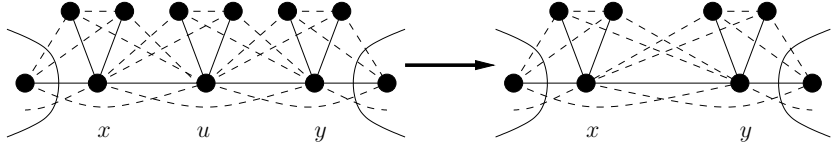}
\caption{Reduction of paths; the edges of $H'$ are shown by the solid lines.
\label{fig:path-informal}}
\end{figure}

Suppose that $H''$ is a square root of $G''$ with at most $n''+k-1$ edges. Let $H^*$ be the graph obtained from $H''$ by deleting all pendant vertices of $H''$. Then $H^*$ has no vertices of degree 1, and the length of every path $P$ with internal vertices of degree 2 in $H^*$ is bounded by a constant. 
This means that the size of $H^*$ is $O(k)$.  
The vertices of $V_{G''}\setminus V_{H^*}$ are pendant vertices of $H''$. Consider the set $Z$ of pendant vertices of $H''$ adjacent to a vertex $u\in V_{H^*}$. Then the vertices of $Z$ are simplicial vertices of $G''$. Moreover, they are true twins. We use Observation~\ref{obs:leaves} and Lemma~\ref{lem:twins} to show that we may reduce the number of true twins in $G''$ if $G''$
has too many. This results in a graph $G'''$ with $n'''$ vertices such that $n'''$ is $O(k^2)$. 

During the reduction from $G$ to $G''$ we label some edges, that is, we include some edges in sets $R$ or $B$ and, therefore, 
obtain an instance $(G''',k,R,B)$ of {\sc Tree $+\;k$ Edges Square Root with Labels}.

\medskip
Before we give a formal description of our reduction, we introduce the following terminology. A square root $H$ of a graph $G$ that has at most $|V_G|-1+k$ edges for some $k\geq 0$ is called
 a {\it solution} of the instance $(G,k)$ of {\sc Tree $+\;k$ Edges Square Root}. If $R\subseteq E_H$ and $B\cap E_H=\emptyset$ for two disjoint subsets $R$ and $B$ of $E_G$, then $H$ is also called 
 a {\it solution} of the instance $(G,k,R,B)$ of {\sc Tree $+\;k$ Edges Square Root with Labels}.
 
We are now ready to give the exact details of our reduction. 
Let $G$ be a connected graph with $n$ vertices and $m$ edges, and let $k$ be a positive integer. First we check whether $G$ has a square root that is a tree by using the linear-time algorithm of  
Lin and Skiena~\cite{LinS95}. If we find such a square root, then we stop and return a yes-answer. From now we assume that every square root of $G$ (if there exists one) has at least one cycle.

Because connected graphs that have square roots are 2-connected, we also check whether $G$ is 2-connected. If so, then we stop and return a no-answer. Otherwise we continue as follows. 
We introduce two sets of edges $R$ and $B$. Initially, we set $R=B=\emptyset$. Next, we ``trim'' pendant edges in potential roots, that is, we exhaustively apply the following rule that consists of five steps that must be performed in increasing order.

\medskip
\noindent
{\bf Trimming Rule}
\begin{enumerate}
\item Find a pair $S=\{u_1,u_2\}$  
of two adjacent vertices such that one connected component of $G-S$ consists of  $r\geq 3$ vertices $u_3,\ldots,u_r$ that together with $u_1,u_2$ 
form a clique in $G$.
\item If $N_G[u_1]=N_G[u_2]$ then stop and return a no-answer.
\item If $N_G[u_1]\setminus N_G[u_2]\neq\emptyset$ and  $N_G[u_2]\setminus N_G[u_1]\neq\emptyset$, then stop and return a no-answer.
\item If $N_G[u_1]\setminus N_G[u_2]\neq\emptyset$, then rename $u_1$ by $u_2$ and $u_2$ by $u_1$ 
(this step is for notational convenience  only and has no further meaning).
\item Define sets $R'=\{u_1u_2,\ldots,u_1u_r\}$ and $B'=\{u_iu_j\; |\; 2\leq i<j\leq r\}\cup\{u_1x\; |\; x\in N_G(u_1)\setminus\{u_2,\ldots,u_r\}\}$.
\item If $R\cap B'\neq\emptyset$ or $R'\cap B\neq\emptyset$, then stop and return a no-answer. Otherwise, 
set $R=R\cup R'$, $B=B\cup B'$,
delete $u_3,\ldots,u_r$ from $G$ and also delete all edges incident to $u_3,\ldots,u_r$  from $R$ and $B$.
\end{enumerate}

Exhaustively applying the trimming rule
yields a sequence of instances $(G_0,k,R_0,B_0),\ldots, (G_\ell,k,R_\ell,B_\ell)$ of {\sc Tree $+\;k$ Edges Square Root with Labels} 
for some integer $\ell\geq 0$,
where $(G_0,k,R_0,B_0)=(G,k,\emptyset,\emptyset)$ and where $(G_\ell,k,R_\ell,B_\ell)$ is an instance for which we have either
returned a no-answer (in steps~2,~3 or~6) or for which there does not exist a set $S$  as specified in step 1. 
For $0\leq i\leq \ell-1$ we denote the sets $R'$ and $B'$ constructed in the $(i+1)$th call of the trimming rule by $R_i'$ and $B'_i$, respectively.
We need the following lemma. 

\begin{lemma}\label{lem:trimming}
The instance $(G_\ell,k,B_\ell,R_\ell)$ has no solution that is a tree, and $G_\ell$ is $2$-connected.
Moreover, $(G_\ell,k,R_\ell,B_\ell)$ has a solution if and only if $(G_0,k,R_0,B_0)$ has a solution.
If the trimming rule returned a no-answer for $(G_\ell,k,R_\ell,B_\ell)$, then $(G_0,k,R_0,B_0)$ has no solution. 
\end{lemma}

\begin{proof}
For $0\leq i\leq \ell$, we use induction to show that 
the graph $G_i$ is $2$-connected and that $(G_i,k,B_i,R_i)$ has no solution that is a tree.
Moreover, for all $1\leq i\leq \ell$, we show that $(G_i,k,R_i,B_i)$ has a solution if and only if $(G_{i-1},k,R_{i-1},B_{i-1})$ has a solution.
Finally, we prove that if the trimming rule returned a no-answer for $(G_\ell,k,R_\ell,B_\ell)$, then $(G_0,k,R_0,B_0)$ has no solution.
  
If $i=0$, then $G_i$ is 2-connected and $(G_0,k,B_0,R_0)$ has no solution that is a tree by our initial assumption (as we had 
preprocessed $G$ with respect to these two properties).
Now suppose that $1\leq i\leq \ell$. By our induction hypothesis, we may assume that $G_{i-1}$ is 2-connected and that 
$(G_{i-1},k,B_{i-1},R_{i-1})$ has no solution that is a tree.

Because the trimming rule applied on $(G_{i-1},k,R_{i-1},B_{i-1})$ yielded a new instance $(G_i,k,R_i,B_i)$, the graph 
$G_{i-1}$ has a pair $S=\{u_1,u_2\}$ of adjacent vertices such that one 
connected component of $G-S$ consists of  vertices $u_3,\ldots,u_r$ that together with $u_1,u_2$ form a clique in $G_{i-1}$. 
Step 6 implies that $G_i=G_{i-1}-\{u_3,\ldots,u_r\}$. 
Because we did not return a no-answer for $(G_{i-1},k,R_{i-1},B_{i-1})$, we find that
$N_{G_{i-1}}[u_1] \subset N_{G_{i-1}}[u_2]$. Hence, $G_{i-1}$ is not a complete graph. Because $G_{i-1}$ is 2-connected, this means that
$G_i$ is 2-connected.
We now show that any solution for $(G_{i-1},k,B_{i-1},R_{i-1})$ corresponds to a solution for $(G_i,k,B_i,R_i)$, and vice versa. 

First suppose that $H_{i-1}$ is an arbitrary solution for $(G_{i-1},k,B_{i-1},R_{i-1})$. Let $N_{G_{i-1}}(u_1)\setminus \{u_2,\ldots,u_r\}=\{x_1,\ldots,x_p\}$.
Because $N_{G_{i-1}}[u_2]\setminus N_{G_{i-1}}[u_1]\neq\emptyset$, we find that $G_{i-1}-\{u_1,u_2\}$ contains at least two connected components.
As $G_{i-1}$ is 2-connected,
this means that
$\{u_1,u_2\}$ is a minimal $\{u_3,\ldots,u_r\},V_{G_{i-1}}\setminus \{u_1,\ldots,u_r\}$-separator of $G_{i-1}$. Hence
we may apply Lemma~\ref{lem:leaves-two}~iii), which tells us that
$u_2u_1,\ldots,u_ru_1\in E_{H_{i-1}}$,  $u_3u_2,\ldots,u_ru_2\notin E_{H_{i-1}}$, and $u_1x_1,\ldots,u_1x_p\notin E_{H_{i-1}}$. 
As $R_i\subseteq R_{i-1}\cup \{u_1u_2\}$ and $B_i\subseteq B\cup \{u_1x_1,\ldots,u_1x_p\}$, this means that
the graph obtained from $H_{i-1}$ by deleting $u_3,\ldots,u_r$ is a solution for $(G_i,k,R_i,B_i)$; in particular note that
$|E_{H_i}|\leq |E_{H_{i-1}}|-(r-3)\leq |V_{G_{i-1}}|-1+k-(r-3)=|V_{G_i}|-1+k$, as required.

Now suppose that $H_i$ is an arbitrary solution for $(G_i,k,R_i,B_i)$. 
Then adding the edges $u_1u_3,\ldots,u_1u_r$ to $H_i$ yields a graph $H$ that is a square root of $G_{i-1}$.
The edges $u_1u_3,\ldots,u_1u_r$ are not in $B_{i-1}$, as they are in the set $R'_{i-1}$ constructed in step 5 and $R'_{i-1}\cap B_{i-1}=\emptyset$ (otherwise the trimming rule would have stopped when processing $(G_{i-1},k,R_{i-1},B_{i-1})$ in step 6).
Now suppose that $R_{i-1}$ contains an edge not in $H$.
By definition of $R_i$, this edge must be between some $u_s$ and $u_t$ with $3\leq s<t\leq r$. 
Then $u_su_t$ belongs to $R_i$, because it  was placed in the set $R_h$ 
for some $h\leq i-1$. In step 4 of the corresponding call of the trimming rule, also one of the edges $u_su_1$ or $u_tu_1$ was placed in $B_h$. Hence either $u_su_1$ or $u_tu_1$ belongs to $B_{i-1}$.  This yields a contraction as both $u_su_1$ and $u_tu_1$ belong to $R'_{i-1}$ and
$R'_{i-1}\cap B_{i-1}=\emptyset$ (otherwise the trimming rule would have stopped when processing $(G_{i-1},k,R_{i-1},B_{i-1})$ in step 6). 
Hence, after observing that $|E_H|=|E_{H_i}|+(r-3)\leq |V_{G_i}|-1+k+(r-3)=|V_{G_{i-1}}|-1+k$,
we conclude that $H$ is a solution for $(G_{i-1},k,R_{i-1},B_{i-1})$.
We observe that $H_i$ cannot be a tree, as this would imply that $H$ is a tree, which 
is not possible as $(G_{i-1},k,R_{i-1},B_{i-1})$ does not have such a solution. 

We are left to show that if the trimming rule returned a no-answer for $(G_\ell,k,R_\ell,B_\ell)$, then $(G_0,k,R_0,B_0)$ has no solution. Due to
the above, this comes down to showing that $(G_\ell,k,R_\ell,B_\ell)$ has no solution.

Suppose that the trimming rule returned a no-answer for $(G_\ell,k,R_\ell,B_\ell)$. Then this must have happened in step~2,~3 or~6, thus after step~1.
Hence, there exists a pair of adjacent vertices $S=\{u_1,u_2\}$ in $G_\ell$, such that one connected component of $G_\ell-S$ has vertex set $\{u_3,\ldots,u_r\}$  and  $\{u_1,\ldots,u_r\}$ is a clique. 

First assume that $S$ is not a separator of $G_\ell$, that is, $G_\ell$ is a complete graph with vertex set
$\{u_1,\ldots,u_r\}$.  Then $N_G[u_1]=N_G[u_2]$ (and the no-answer given by the trimming rule happens in step 2).
In order to obtain a contradiction, assume that $(G_\ell,k,R_\ell,B_\ell)$ has a solution $H$.
Any star on $|V_{G_\ell}|$ vertices is a square root of $G_\ell$ with at most $|V_{G_\ell}|-1+k$ edges.
However, $H$ cannot be such a star, as $(G_\ell,k,R_\ell,B_\ell)$ has no solution that is a  tree.
Hence, $R_\ell\neq\emptyset$ or $B_\ell\neq\emptyset$.
Recall that $B_0=R_0=\emptyset$. Hence, $\ell\geq 1$, and non-emptiness of $R_\ell$ or $B_\ell$ must have been obtained in a previous call
of the trimming rule, say in the $(h+1)$th call of the trimming rule for some $0\leq h\leq \ell-1$.
By definition of steps~5 and~6, we find that $B_h\neq \emptyset$ implies that $R_h\neq \emptyset$. Hence, $R_h\neq \emptyset$.
Let $u_iu_j\in R_h$. By steps~5 and~6, this edge has an end-vertex, say $u_i$, such that 
$u_iu_s\in B_\ell$ for all
$s\in\{1,\ldots,r\}\setminus \{i,j\}$. Consequently, $u_ju_s\in E_H$ for all $s\in\{1,\ldots,r\}\setminus \{j\}$.   
Because the star with central vertex $u_j$ and  leaves $V_{G_\ell}\setminus\{u_j\}$ is not a solution for $(G_\ell,k,R_\ell,B_\ell)$, there must be an edge $u_su_t\in R_\ell$ with $s,t\in\{1,\ldots,r\}\setminus \{j\}$. However then, due to steps~5 and~6, $u_ju_s\in B_\ell$ or $u_ju_t\in B_\ell$, that is, at least one of these edges cannot be in $H$; a contradiction.

Now assume that $S$ is a separator of $G_\ell$. Because $G_\ell$ is 2-connected, both $u_1$ and $u_2$ have at least one neighbor in
 $V_{G_\ell}\setminus\{u_1,\ldots,u_r\}$. Hence $\{u_1,u_2\}$ is a minimal separator (and we may apply Lemma~\ref{lem:leaves-two} in the remainder).
Recall that the trimming rule only returns a no-answer in steps~2,~3, or~6. We consider each of these three cases separately.

\medskip
\noindent
{\bf Case 1.} The no-answer is given in step 2.
Then $N_G[u_1]=N_G[u_2]$. By Lemma~\ref{lem:leaves-two} i) and ii), 
$G_\ell$ is the union of two  
cliques $\{u_1,\ldots,u_r\}$ and $\{u_1,u_2,x_1, \ldots x_p\}$ where $\{x_1,\ldots,x_p\}=N_G(u_1)\setminus\{u_2,\ldots,u_r\}$.
In order to obtain a contradiction, suppose that $(G_\ell,k,R_\ell,B_\ell)$ has a solution~$H$.
By Lemma~\ref{lem:leaves-two}~i) and~ii), we may assume without loss of generality that
$u_1u_2,\ldots,u_1u_r\in E_H$, $u_2u_3,\ldots,u_2u_r\notin E_H$,
$u_1x_1,\ldots,u_1x_p\notin E_H$ and $u_2x_1,\ldots,u_2x_p\in E_H$. Recall that $(G_\ell,k,R_\ell,B_\ell)$ has no solution that is a tree.
Hence, there exists an edge $u_iu_j\in R_\ell$ for some $i,j\in \{2,\ldots,r\}$ or 
an edge $x_ix_j\in R_\ell$ for some $i,j\in \{x_1,\ldots,x_p\}$. By symmetry, we only need to consider the case $u_iu_j\in R_\ell$.
This edge was placed in $R_\ell$ in some previous call of the trimming rule.  However, due to steps~5 and~6 performed in that call,
we find that $u_iu_1\in B_\ell$ or $u_ju_1\in B_\ell$, that is, at least one of these two edges cannot be in $H$; a contradiction.

\medskip
\noindent
{\bf Case 2.} The no-answer is given in step 3.
Then we have $N_G[u_1]\setminus N_G[u_2]\neq\emptyset$ and  $N_G[u_2]\setminus N_G[u_1]\neq\emptyset$.
Due to Lemma~\ref{lem:leaves-two}~i) and~iii), $(G_\ell,k,R_\ell,B_\ell)$ has no solution.

\medskip
\noindent
{\bf Case 3.} The no-answer is given in step 6. Then $R_\ell\cap B_\ell'\neq \emptyset$ or $R'_\ell\cap B_\ell\neq \emptyset$.
By step 4, we may assume that $N_G[u_1]\setminus N_G[u_2]=\emptyset$ and that $N_G[u_2]\setminus N_G[u_1]\neq\emptyset$. 
In order to obtain a contradiction, suppose that $(G_\ell,k,R_\ell,B_\ell)$ has a solution~$H$.
By Lemma~\ref{lem:leaves-two}~iii), $R'_\ell=\{u_2u_1,\ldots,u_ru_1\} \subseteq E_H$. Hence $R'_\ell\cap B_\ell=\emptyset$, which means that
$R_\ell\cap B_\ell'\neq \emptyset$.

Let $\{x_1,\ldots,x_p\}=N_G(u_1)\setminus\{u_1,\ldots,u_r\}$. Then we have that $B_\ell'=\{u_iu_j\; |\; 2\leq i<j\leq r\}\cup\{u_1x_1,\ldots,u_1x_p\}$.
By the same arguments as used in Case 1, we find that $u_iu_j\notin R_\ell$ for all $2\leq i<j\leq r$.
 By Lemma~\ref{lem:leaves-two}~iii),  we find that $E_H$, and hence $R_\ell$, does not contain the edges $u_1x_1,\ldots,u_1x_p$.
We conclude that $R_\ell\cap B_\ell'=\emptyset$; a contradiction. 
\end{proof}

Lemma~\ref{lem:trimming} shows that the trimming rule is safe, that is, we either found that $(G,k,\emptyset,\emptyset)$ has no solution, or that 
we may continue with the instance $(G_\ell,k,R_\ell,B_\ell)$ instead.
Suppose the latter case holds. 
Recall that  $(G_\ell,k,R_\ell,B_\ell)$ has no set $S$ as specified in step~1, as otherwise we would have applied the trimming rule once more.

To simplify notation, we write $(G,k,R,B)=(G_\ell,k,R_\ell,B_\ell)$. We need the following properties that hold for every solution of $(G,k,R,B)$ (should 
$(G,k,R,B)$ have a solution).

\begin{lemma}\label{lem:trimmed}
Any solution $H$ of $(G,k,R,B)$ satisfies the following properties:
\begin{itemize}
\item[i)] the  neighbor of every pendant vertex of $H$ has 
at least two non-pendant neighbors in $H$;
\item[ii)] only edges of $G$ incident to pendant vertices of $H$ can be in $R$ or $B$;  
\item[iii)] if a pendant vertex $v$ of $H$ is incident to an edge of $R$ in $G$, then all other edges of $G$ that are incident to $v$ are in $B$.
\end{itemize} 
\end{lemma}

\begin{proof}
In order to show i), suppose that $H$ is a solution of an instance $(G,k,R,B)$, such that $H$ contains a pendant vertex $u$ adjacent to a vertex $v$. If $d_H(v)=1$, then $H$ is isomorphic to $K_2$, which is not possible as $(G,k,R,B)$ has no solution that is a tree. Hence $d_H(v)\geq 2$ and $v$ has at least one neighbor other than $u$. If all  neighbors of $v$ are pendant, then $H$ is a tree; a contradiction. Hence, $v$ has at least one non-pendant neighbor.
If $v$ has a unique non-pendant neighbor $w$, then 
by Lemma~\ref{lem:leaves-one}, 
$G-\{v,w\}$ contains a connected component induced by the pendant neighbors of $v$ 
whose vertices together with $v$ and $w$ form a clique in $G$. Hence, 
we can apply the trimming rule on $S=\{v,w\}$, which is a contradiction. 
Properties~ii) and~iii) follow from the construction of $R$ and $B$ in steps~4 and~5 of the 
trimming rule. 
\end{proof}
 
We now exhaustively apply the following rule on $(G,k,R,B)$. This rule consists of four steps that must be performed in increasing order. 

\medskip
\noindent
{\bf Path Reduction Rule}
\begin{enumerate}
\item Find an $F$-triple $S=\{u_1,u_2,u_3\}$.
\item Set $R'=\{u_2u_1,u_2u_3,\ldots,u_2u_r\}$ and $B'=\{x_1u_2,\ldots,x_pu_2\}\cup\{y_1u_2,..,y_qu_2\}\cup\{u_1u_3,\ldots,u_1u_r\}\cup\{u_3u_4,\ldots,u_3u_r\}$ (note that the set $\{u_3u_4,\ldots, u_3u_r\}=\emptyset$ if $r=3$).
\item If $R\cap B'\neq\emptyset$ or $R'\cap B\neq\emptyset$, then stop and return a no-answer. 
\item Delete  $u_2,u_4,\ldots,u_r$ from $G$. Delete all edges incident to $u_2,u_4,\ldots,u_r$ from $R$ and $B$. If $u_1u_3\in B$, then delete $u_1u_3$ from $B$. Add $u_1u_3$ to $R$. 
Add  $x_1u_3,\ldots,x_pu_3$ and $y_1u_1,\ldots,y_qu_1$ in $G$. Put these edges in $B$.
\end{enumerate}

Exhaustively applying the path reduction rule
yields a sequence of instances $(G_0,k,R_0,B_0),\ldots, (G_\ell,k,R_\ell,B_\ell)$ of {\sc Tree $+\;k$ Edges Square Root with Labels} 
for some integer $\ell\geq 0$,
where $(G_0,k,R_0,B_0)=(G,k,R,B)$ and where $(G_\ell,k,R_\ell,B_\ell)$ is an instance for which we have either
returned a no-answer (in step~3) or for which there does not exist an $F$-triple $S$. 
For $0\leq i\leq \ell$ we denote the sets $R'$ and $B'$ constructed in the $(i+1)$th call of the path reduction rule by $R_i'$ and $B'_i$, respectively.

We need the following lemma, which we will use at several places.

\begin{lemma}\label{l-back}
Let $1\leq i\leq \ell$ and $\{u_1,u_2,u_3\}$ be the $F$-triple that yielded instance
$(G_i,k,R_i,B_i)$. If $H_i$ is a solution for $(G_i,k,R_i,B_i)$, then $u_1u_3\in E_{H_i}$ and the graph $H_{i-1}$ obtained from $H_i$ by removing 
the edge $u_1u_3$ and by adding $u_2$ and vertices $u_4,\ldots,u_r$ (if $r\geq 4$) together with edges $u_2u_1,u_2u_3,\ldots,u_2u_r$ 
is a solution for $(G_{i-1},k,R_{i-1}, B_{i-1})$.
\end{lemma}

\begin{proof}
We find that $u_1u_3$ is an edge in $H_i$, because $u_1u_3\in R_i$ due to step 4 of the last call
of the path reduction rule. The graph $H_{i-1}$ is not only a square root of $G_{i-1}$ but even  a solution for $(G_{i-1},k,R_{i-1},B_{i-1})$ for the following reasons.
First, $H_{i-1}$ has at most $|V_{G_{i-1}}|-1+k$ edges. Second, $H_{i-1}$ contains no edge of $B_{i-1}$ as the added edges 
$u_2u_1,u_2u_3,\ldots,u_2u_r$  are all in $R_{i-1}'$
and $R_{i-1}'\cap B_{i-1}=\emptyset$. Third, $H_{i-1}$ contains all the edges of $R_{i-1}$, which can be seen as follows.
Suppose that $H_{i-1}$ misses an edge of $R_{i-1}$. Then this edge must be in  $\{x_1u_2,\ldots,x_pu_2\}\cup\{y_1u_2,..,y_qu_2\}\cup\{u_1u_3,\ldots,u_1u_r\}\cup\{u_3u_4,\ldots,u_3u_r\}$. However, this set is equal to $B'_{i-1}$ and $R_{i-1}\cap B_{i-1}'=\emptyset$.
\end{proof}

We also need the following lemma about true twins in  $G_0,\ldots,G_\ell$ that we will use later as well.

\begin{lemma}\label{l-twins}
Let $1\leq i\leq \ell$ and $\{u_1,u_2,u_3\}$ be the $F$-triple that yielded instance
$(G_i,k,R_i,B_i)$.
Then any true twins $v,w\in V_{G_i}\setminus\{u_1,u_3\}$ in $G_i$ are true twins in $G_{i-1}$.
\end{lemma}

\begin{proof}
Suppose that $G_i$ has true twins  $v,w\in V_{G_i}\setminus\{u_1,u_3\}$ that are not true twins in $G_{i-1}$. Consider the corresponding $F$-graph that yielded the instance  $(G_i,k,R_i,B_i)$. 
Because $v,w$ are not true twins in $G_{i-1}$, the neighborhood of $v$ or $w$ is modified by the path reduction rule. We may assume without loss of generality that the neighborhood of $v$ is changed.
Note that neither $v=u_2$ nor $v\in \{u_4,\ldots,u_r\}$ if $r\geq 3$, because these vertices have been removed in step 4 of the path reduction rule when $G_i$ was constructed. 
As $v\notin \{u_1,u_3\}$ either, we find that $v\in\{x_1,\ldots,x_p\}\cup\{y_1,\ldots,y_q\}$.  By symmetry we may assume that $v\in\{x_1,\ldots,x_p\}$. We observe that $v$ is adjacent to both $u_1$ and $u_3$ in $G_i$.
 Because  the neighborhood of each $x_i$ is modified in the same way (namely by the removal of $u_2$ and the addition of $u_3$), we find that  $w\notin\{x_1,\ldots,x_p\}$. Because $u$ and $v$ are true twins, they are adjacent. Because no two vertices $x_i$ and $y_j$ are adjacent in $G_i$, we then obtain that $w\notin\{y_1,\ldots,y_q\}$. We conclude that  the neighborhood of $w$ is not modified by the application of the path reduction rule. Because $v$ is adjacent to $u_1$ and $u_3$ in $G_i$ and $v,w$ are true twins in $G_i$, this means that $w$ is adjacent to $u_1$ and $u_3$ in $G_{i-1}$ already. However, by definition of an $F$-graph, $N_{G_{i-1}}(u_1)\cup N_{G_{i-1}}(u_3)=\{u_2,u_4,\ldots,u_r\}$, and $u_2,u_4,\ldots,u_r$ are not in $G_i$ as they were removed by the path reduction rule; a contradiction.   
\end{proof}

The next lemma is the analog of Lemma~\ref{lem:trimming} for the path reduction rule.

\begin{lemma}\label{l-reduction}
The instance $(G_\ell,k,B_\ell,R_\ell)$ has no solution that is a tree, and $G_\ell$ is $2$-connected.
Moreover, $(G_\ell,k,R_\ell,B_\ell)$ has a solution if and only if $(G_0,k,R_0,B_0)$ has a solution.
If the path reduction rule returned a no-answer for $(G_\ell,k,R_\ell,B_\ell)$, then $(G_0,k,R_0,B_0)$ has no solution. 
\end{lemma}

\begin{proof}
For $0\leq i\leq \ell$, we use induction to show that 
the graph $G_i$ is $2$-connected and that $(G_i,k,B_i,R_i)$ has no solution that is a tree.
Moreover, for all $1\leq i\leq \ell$, we show that $(G_i,k,R_i,B_i)$ has a solution if and only if $(G_{i-1},k,R_{i-1},B_{i-1})$ has a solution.
Finally, we prove that if the path reduction rule returned a no-answer for $(G_\ell,k,R_\ell,B_\ell)$, then $(G_0,k,R_0,B_0)$ has no solution.
  
If $i=0$, then $G_i$ is 2-connected and $(G_0,k,B_0,R_0)$ has no solution that is a tree by Lemma~\ref{lem:trimming}.
Now suppose that $1\leq i\leq \ell$. By our induction hypothesis, we may assume that $G_{i-1}$ is 2-connected and that 
$(G_{i-1},k,B_{i-1},R_{i-1})$ has no solution that is a tree.

Because the path reduction rule applied on $(G_{i-1},k,R_{i-1},B_{i-1})$ yielded a new instance $(G_i,k,R_i,B_i)$, the graph 
$G_{i-1}$ has an $F$-triple $S=\{u_1,u_2,u_3\}$. Because $G_{i-1}$ is 2-connected, $G_i$ is 2-connected; in particular note that $p\geq 1$ and 
$q\geq 1$ by definition of an $F$-triple.

First suppose that $H_{i-1}$ is a solution for $(G_{i-1},k,R_{i-1},B_{i-1})$. 
We claim that $H_{i-1}$ contains no edge $u_su_t\in R_{i-1}$ with $4\leq s<t\leq r$. We prove this claim by contradiction: let $u_su_t\in E_{H_{i-1}}\cap R_{i-1}$ for some $4\leq s<t\leq r$.

Suppose that $u_su_t\in R_0$.
We may apply Lemma~\ref{lem:trimmed} as $(G_0,k,R_0,B_0)$ has a solution $H_0$; if 
$i\geq 1$ this fact follows from the induction hypothesis.
By Lemma~\ref{lem:trimmed} we find that either $u_s$ is a pendant vertex in $H_0$ with $u_t$ as its (unique) neighbor, or the other way around.
We may assume without loss of generality that the first case holds, that is, $u_s$ is pendant in $H_0$ and has $u_t$ as its neighbor.
Note that $N_{G_0}[u_s]\subseteq N_{G_0}[u_t]$.
We claim that $N_{G_h}[u_s]\subseteq N_{G_h}[u_t]$ for all $0\leq h\leq i-1$.
To obtain a contradiction, suppose not. Then at some point $u_s$ will be made adjacent to a vertex $v$ not adjacent to $u_t$ for the first time 
in step 4 of  some call of the path reduction rule.
Let $S=\{u_1',u_2',u_3'\}$ be the corresponding $F$-triple. Then we may assume without loss of
generality that either $u_s\notin \{u_1',u_2',u_3'\}$ is adjacent to $u_1'$ and $u_2'$ but not to $u_3'=v$, or that $v\notin \{u_1',u_2',u_3'\}$
is adjacent to $u_1',u_2'$ but not to $u_3'=u_s$. In the first case, $u_t$ is not in $\{u_1',u_2'\}$, but must be adjacent to $u_1'$ and $u_2'$ by our assumption, and hence,  the edge $u_tu_3'=u_tv$ will be added in the same step; a contradiction. In the second case, as $u_s$ is adjacent to $u_1'$ and $u_2'$, also $u_t$ is adjacent to $u_1'$ and $u_2'$ (again by our assumption). Because $u_t$ does not get removed in this step (as $u_t$ belongs to $G_{i-1}$), this violates the definition of an $F$-triple. We conclude that  $N_{G_h}[u_s]\subseteq N_{G_h}[u_t]$ for all $0\leq h\leq i-1$.

We first assume that $u_su_2$ is an edge in $G_0$. Step~4 of the path reduction rule only moves an edge $u_1'u_3'$ from a $B$-set to an $R$-set if
$u_1'$ and $u_3'$ are outer vertices of an $F$-triple. In that case all their common neighbors will be removed from the graph by the definition of an $F$-triple. Because  $N_{G_h}[u_s]\subseteq N_{G_h}[u_t]$ for all $0\leq h\leq i-1$, we find that $u_t$ is a common neighbor of $u_2$ and $u_s$ in $G_h$
for all $0\leq h\leq i-1$; in particular $u_t$ belongs to $G_{i-1}$.
Hence, the edge $u_su_2$ will never be moved
from $B_h$ to $R_h$ in step~4 of the $(h+1)$th call of the path reduction rule  for some $0\leq h\leq i-1$.
If $u_su_2$ is not an edge in $G_0$, then at some point it will be an edge due to step~4 of some call of the path reduction rule, say the $(h^*+1)$th call for some $0\leq h^*\leq i-1$. In the same step, $u_su_2$ will be placed in the set $B_{h^*}$. Then, again because $N_{G_h}[u_s]\subseteq N_{G_h}[u_t]$ for all $0\leq h\leq i-1$, the edge $u_su_2$ will never be moved from $B_{h^*}$ to a set $R_h$ for some $h^*< h\leq i-1$. 
Hence, in both cases, we find that $u_su_2\in B_{i-1}$ even if $i\geq 1$.
As $u_su_2\in R'_{i-1}$ (due to step~2 in the $i$th call), we find that $R'_{i-1}\cap B_{i-1}\neq \emptyset$.
Hence, the path reduction rule would return a no-answer
for  $(G_{i-1},k,R_{i-1},B_{i-1})$ in step~3, and consequently the instance $(G_i,k,R_i,B_i)$ would not exist; a contradiction. 

Now suppose that $u_su_t$ was placed in some set $R_h$ for some $1\leq h\leq i-1$. Properties~ii) and~iii) of an $F$-graph
together with step 4 of the path reduction rule imply the following: if $u_s$ and $u_t$ form a triangle with some vertex $z$,  
then $u_sz\in B_h$ or $u_tz\in B_h$.  Moreover, in the case in which $z\in V_{G_{i-1}}$, this property is not violated by any subsequent intermediate calls of the path reduction rule. 
Hence, if $u_su_t\in R_{i-1}$, then $u_su_2\in B_{i-1}$ or $u_tu_2\in B_{i-1}$, and as $\{u_su_2,u_tu_2\}\subseteq R'_{i-1}$ as well,
we derive the same contradiction as before.
We conclude that $H_{i-1}$ contains no edge $u_su_t\in R_{i-1}$ with $4\leq s<t\leq r$. 
Also, by Lemma~\ref{lem:path-two}, we may assume without loss of generality that $H_{i-1}$ contains no edge $u_su_t\notin R_{i-1}$ with $4\leq s<t\leq r$; otherwise we could remove such an edge from $H_{i-1}$, and the resulting graph would still be a solution for $(G_{i-1},k,R_{i-1},B_{i-1})$.
Consequently, $u_4,\ldots,u_r$ are pendant vertices of $H_{i-1}$.  This means that the graph $H$ obtained from $H_{i-1}$ by deleting vertices
$u_2,u_4,\ldots,u_r$ and adding the edge $u_1u_3$ is not only a square root of $G_i$ with at most 
$|V_{G_i}|-1+k$ edges but even a solution for $(G_i,k,R_i,B_i)$.

Now suppose that $H_i$ is a solution for $(G_i,k,R_i,B_i)$. 
By Lemma~\ref{l-back}, the graph $H$ obtained from $H_i$ by removing 
the edge $u_1u_3$ and by adding $u_2$ and vertices $u_4,\ldots,u_r$ (if $r\geq 4$) together with edges $u_2u_1,u_2u_3,\ldots,u_2u_r$ 
is a solution for $(G_{i-1},k,R_{i-1}, B_{i-1})$.
We observe that $H_i$ cannot be a tree, as this would imply that $H$ is a tree, which 
is not possible as $(G_{i-1},k,R_{i-1},B_{i-1})$ does not have such a solution by the induction hypothesis. 

Finally, suppose that  the path reduction rule returned a no-answer for $(G_\ell,k,R_\ell,B_\ell)$. We must show that $(G_0,k,R_0,B_0)$ has no solution. Due to the above this comes down to showing that $(G_\ell,k,R_\ell,B_\ell)$ has no solution.
The only step in which the path reduction rule can return a no-answer is in step~3, meaning that $G_\ell$ has an $F$-triple $S=\{u_1,u_2,u_3\}$ such that $R_\ell\cap B_\ell'\neq\emptyset$ or $R_\ell'\cap B_\ell\neq\emptyset$.

In order to obtain a contradiction, suppose that $(G_\ell,k,R_\ell,B_\ell)$ has a solution $H$.
By Lemma~\ref{lem:path-two}, the graph $F$ shown in Figure~\ref{fig:path} is a subgraph of $H$ such that $d_H(u_2)=r-1$, 
$\{x_1,\ldots,x_p\}=N_H(u_1)\setminus\{u_2\}$ and  $\{y_1,\ldots,y_q\}=N_H(u_3)\setminus\{u_2\}$. Consequently, $R'_\ell=\{u_2u_1,u_2u_3,\ldots,u_2u_r\}\subseteq E_H$, and hence $R'_\ell\cap B_\ell=\emptyset$, and moreover, $E_H\cap B'_\ell= E_H\cap (\{x_1u_2,\ldots,x_pu_2\}\cup\{y_1u_2,..,y_qu_2\}\cup\{u_1u_3,\ldots,u_1u_r\}\cup\{u_3u_4,\ldots,u_3u_r\})=\emptyset$, and hence $R_\ell\cap B'_\ell=\emptyset$; a contradiction.
\end{proof}
 
Lemma~\ref{l-reduction} shows that the path reduction rule is safe, that is, we either found that $(G_0,k,R_0,B_0)$ has no solution, or that 
we may continue with the instance $(G_\ell,k,R_\ell,B_\ell)$ instead.
Suppose the latter case holds. 
Recall that $(G_\ell,k,R_\ell,B_\ell)$ has no $F$-triple, as otherwise we would have applied the path reduction rule once more.
Also recall that $R_0$ is the set of vertices in the set $R$ immediately after the trimming rule.
We write $R^1=R_0\cap R_\ell$ and $R^2=R_\ell\setminus R_0$. To simplify notation, from now on, we also write $(G,k,R,B)=(G_\ell,k,R_\ell,B_\ell)$; note that $R=R^1\cup R^2$. We need the following properties  that hold for every solution of $(G,k,R,B)$ (should $(G,k,R,B)$ have a solution).
 
We call an induced cycle $C$ in a graph $H$ {\it semi-pendant} if all but at most one of the vertices of $C$  are only adjacent to pendant vertices of $H$ and their neighbors on~$C$.
Similarly, we call an induced path $P$ in a graph $H$ {\it semi-pendant} if all internal vertices of $P$ are only adjacent to pendant vertices of~$H$ and their neighbors on~$P$.

\begin{lemma}\label{lem:reduced}
Any solution $H$ of $(G,k,R,B)$ has the following properties:
\begin{itemize}
\item[i)] the neighbor of every pendant vertex of $H$ has 
at least two non-pendant neighbors in $H$;
\item[ii)] only edges of $G$ incident to pendant vertices of $H$ can be in $R^1$, 
and  if a pendant vertex $v$ of $H$ is incident to an edge of $R$, then 
all other edges of $G$ that are incident to $v$ are in $B$;
\item[iii)] no edge of $R^2$ is  incident to a pendant vertex of $H$;
\item[iv)] the length of every semi-pendant path  in $H$ is at most~$5$;
\item [v)] the length of every semi-pendant cycle in $H$ is at most~$6$.
\end{itemize} 
\end{lemma}

\begin{proof}
We prove that property~i) holds by  contradiction. Suppose that $H$ contains a vertex $v$ that is the (unique) neighbor of a pendant vertex $u$, such that $v$ has at most one non-pendant neighbor in $H$.
If all neighbors of $v$ in $H$ are pendant, then $H$ is a tree. However, this would contradict Lemma~\ref{l-reduction}. Hence, $v$ has a unique non-pendant neighbor in $H$. 
Recall that $H$ is a solution for $(G_\ell,k,R_\ell,B_\ell)$. Note that if $v$ is an outer vertex of the corresponding $F$-triple, then Lemma~\ref{l-back}
tells us that $(G_{\ell-1},k,R_{\ell-1},B_{\ell-1})$ has a solution $H_{\ell-1}$ in which $v$ is a non-pendant vertex that has at least one pendant neighbor and that has a unique non-pendant neighbor. Hence, by applying Lemma~\ref{l-back} inductively, we obtain that $(G_0,k,R_0,B_0)$ has a solution $H_0$ containing a vertex with exactly the same property. This contradicts Lemma~\ref{lem:trimmed}~i). We conclude that property~i) holds.

We now show property~ii). By Lemma~\ref{lem:trimmed}, every edge of $G_0$ that is in $R_0$ is incident to a pendant vertex $u$ of any solution for
$(G_0,k,R_0,B_0)$ such that all the other edges of $u$ belong to $B_0$. We observe that, when applying the path reduction rule, $u$ will neither be in an $F$-triple nor removed from the graph, but $u$ could be a vertex of $x$-type or $y$-type. Hence, the path reduction rule may change the neighbors of $u$ but if so any new edges incident to it will be placed in $B$ (and stay in $B$ afterward). Consequently, $u$ must be a pendant vertex in any solution for
$(G,k,R,B)(=(G_\ell,k,R_\ell,B_\ell)$) as well. We conclude that~ii) holds.

We now prove prove property~iii). Recall that we applied the path reduction rule only after first applying the trimming rule exhaustively.
When we apply the path reduction rule on an $F$-triple $\{u_1,u_2,u_3\}$, then afterward 
$u_1$ and $u_3$ have degree at least~2 in any solution
for the resulting instance, which can be seen as follows. The edge $u_1u_3$ is added to $R^2\subseteq R$, and hence belongs to any solution. We also have that $u_1$ is adjacent to $x_1$ in $G$, whereas the edge $u_3x_1$ belongs to $B$. 
This means that $u_1$ cannot be made adjacent to $x_1$ via the path $u_1u_3x_1$ in $H$, and as such must have at least one other neighbor in $H$. For the same reason $u_3$, which is adjacent to $y_1$ in $G$ whereas $u_1y_1\in B$, must have another neighbor in 
$H$ besides $u_1$. As a consequence, any edge in $R^2$ cannot be incident to a pendant vertex of $H$, that is, we have shown property~iii). 

We now prove property~iv). Let~$P$ be a semi-pendant path of length at least~6 in~$H$. By definition, $P$ is an induced path. Hence, we can take any three consecutive vertices of $P$ as the three vertices $u_1,u_2,u_3$ in Lemma~\ref{lem:path-one}. By applying this lemma, we find that $G$ is an $F$-graph implying that we could have applied
the path reduction rule once more; a contradiction. Property~v) can be proven by using the same arguments.
\end{proof}

We need the following lemma that holds in case a solution exists for $(G,k,R,B)$.

\begin{lemma}\label{l-hstarbounded}
The number of non-pendant vertices of any solution for $(G,k,R,B)$ is at most $15k-14$.
\end{lemma}

\begin{proof}
Suppose $(G,k,R,B)$ has a solution $H$.
Let $Z$ be the set of pendant vertices of $H$, and let $H^*=H-Z$. 
We need to show that $V_{H^*}$ has at most $15k-14$ vertices.
Let $V'$ be the set of vertices that have degree at least $3$ in $H^*$, and let $V''$ be the set of vertices of degree 2 in $H^*$.
By Lemma~\ref{lem:reduced} i) every vertex of $H$ that is adjacent to a pendant vertex of $H$ has degree 
at least $2$ in $H^*$. 
Hence, $H^*$ has no vertices of degree at most 1, that is, $V_{H^*}=V'\cup V''$. 
 Because $H$ is a solution for $(G,k,R,B)$, we have that $|E_H|\leq |V_G|-1+k=|V_H|-1+k$. This means that
\[\begin{array}{lcl}
|V'|+|V''|-1+k &=       &|V_H|-|Z|-1+k\\[1mm]
                       &\geq &|E_H|-|Z|\\[1mm]
                       &=       &|E_{H^*}|\\[1mm]
                       &=      &\frac{1}{2}\sum_{v}d_{H^*}(v)\\[2mm]
                      &\geq  &\frac{1}{2}(3|V'|+2|V''|).
\end{array}\]                      
Hence, $|V'|\leq 2k-2$. Let $\alpha$ be the number of paths in $H^*$ that only have internal vertices of degree~2;
note that by Lemma~\ref{lem:reduced}~iv) the length of such paths is at most~5. 
Let $\beta$ be the number of cycles in $H^*$ that have exactly one vertex of degree at least~3;
note that by Lemma~\ref{lem:reduced}~v) the length of such cycles is at most~6.
Because $|E_{H^*}|\leq |V'|+|V''|-1+k$, we find that $\alpha+\beta\leq 2k-2-1+k=3k-3$ and that $\beta\leq k$.
Hence, $|V''|\leq 5k+4((3k-3)-k)=13k-12$. Consequently, $H^*$ has at most $2k-2+13k-12=15k-14$ vertices. 
\end{proof}

We are now ready to state our final reduction rule. The goal of this rule is to apply it once in order to deduce either that $(G,k,R,B)$ has no solution or to derive a new instance of bounded size.
A {\it true twin partition} of a set of vertices $S$ of a graph $G$ is a partition $S_1,\ldots,S_t$ of $S$
such that for all $u,v\in S$ and all $1\leq i\leq t$ we have that $u$ and $v$ are in $S_i$ if and only if $u$ and $v$ are true twins in $G$.
If $S$ consists of simplicial vertices only we observe that there is no edge between any two vertices that belong to different sets $S_i$ and $S_j$.

\medskip
\noindent
{\bf Simplicial Vertex Reduction Rule}
\begin{enumerate}
\item Find the set  $S$  of all simplicial vertices of $G$ that are not incident to the edges of $R^2$, and moreover, that have all but one of their incident edges in $B$ should they be incident to an edge of $R^1$.
\item If $|V_G\setminus S|>15k-14$, then stop and return a no-answer.
\item Construct the true twin partition $S_1,\ldots,S_t$ of $S$.
Let $X_1,\ldots,X_t$ be the sets of vertices incident to an edge of $R^1$ in $S_1,\ldots,S_t$, respectively.
\item If $t> 15k-14$, then stop and return a no-answer.
\item If there exist a set $X_i$ such that the edges of $R^1$ incident to a vertex of $X_i$ have no common end-vertex,
 then stop and return a no-answer.
\item If there exist  a set $S_i$ such that $|S_i\setminus X_i|\geq 15k-13$ and such that there are 
three vertices $u\in X_i$, $v\in N_G(u)$ and $x\in S_i\setminus X_i$ with 
$uv\in R^1$ and $xv\in B$, then stop and return a no-answer.
\item For $i=1,\ldots,t$, if $|X_i|>1$, then take $|X_i|-1$ arbitrary vertices of $X_i$ and delete them both from $G$ and  from $S_i$, also delete the edges of $R$ and $B$ that are incident to these vertices.
\item For $i=1,\ldots,t$, if $|S_i|>15k-13$, then delete $|S_i|-15k+13$ arbitrary vertices of $S_i\setminus X_i$ from $G$, also delete the edges of 
$R$ and $B$ that are incident to these vertices.
\end{enumerate}

Applying the simplicial vertex reduction rule on $(G,k,R,B)$ either yields a no-answer (in step 2,~4,~5 or~6) or a new instance $(\hat{G},k,\hat{R},\hat{B})$ of
{\sc Tree $+\;k$ Edges Square Root with Labels}. We will show that if $\hat{G}$ exists, then its size is bounded by a quadratic function of~$k$.  For doing so we first need the following two lemmas. 

\begin{lemma}\label{l-claim1}
For $i=1,\ldots,t$, no vertex of $S_i\setminus X_i$ is incident to an edge~in~$R$.
\end{lemma}

\begin{proof}
By definition of $S_i$, no vertex of $S_i$, and hence no vertex of $S_i\setminus X_i$, 
is incident to an edge in $R^2$.
By definition of $X_i$, no vertex in $S_i\setminus X_i$ is incident to an edge in $R^1$. Because $R=R^1\cup R^2$, we have proven Lemma~\ref{l-claim1}.
\end{proof}

For  $x\in V_G$, we let $B(x)$ denote the set of edges of $B$ incident to $x$.

\begin{lemma}\label{l-claim2}
$B(x)=B(y)$ for all $x,y\in S_i\setminus X_i$.
\end{lemma}

\begin{proof}
Let $x,y\in S_i\setminus X_i$ and let $xz\in B$ for some $z\in V_G$. We first show that $y\neq z$ and
we then prove that $yz\in B$. 

In order to obtain a contradiction, assume that $y=z$.
Then $xy$ was included in $B$ either by an application of the trimming rule or by an application of the path reduction rule. In both cases, $xy$ was also made  adjacent to an edge of $R$. This edge may be deleted later on. Deleting an edge $e$ from 
$R$ happens either in step~6 of the trimming rule or in step~4 of the path reduction rule. 
However, both rules add a new edge $e'$ to $R$ that is adjacent to all the edges that were previously adjacent to $e$ and that were not deleted by the two rules. Hence, $xy$ is still adjacent to an edge of $R$ in $G$. In other words, $x$ or $y$ is incident to an edge of $R$ in $G$. Because $x$ and $y$ belong to $S_i\setminus X_i$, this is not possible due to Lemma~\ref{l-claim1}. Hence, $y\neq z$.

In order to show that $yz\in B$, we again use the observation that whenever the trimming or path reduction rule deletes an edge $e\in R$, the rule adds a new edge $e'$ in $R$ such that $e'$ is adjacent to all the edges $uv$ that were previously adjacent to $e$ and that were not deleted by the rules.  In this case we make the extra
observation that if a vertex $u$ is an end-vertex of $e$
that is not deleted by the rule, 
then $u$ is an end-vertex of $e'$. Because 
the vertices in $S_i\setminus X_i$ are not incident to any edges in $R$ by Lemma~\ref{l-claim1}, we find that $z$
was incident to an edge of $R$ after applying the trimming rule or path reduction rule that added the edge
$xz$ to $B$. 
We also observe that an edge in $B$ is only deleted from $B$ if one of its end-vertices is deleted unless it is added to $R$ by the path reduction rule.
This means that we can argue as follows.
 
First suppose that $xz$ was added to $B$ due to an application of the trimming rule. 
If $y$ was adjacent to $z$ when the rule was applied, then $yz$ was included in $B$ as well by the definition of this rule. 
If  $y$ was made adjacent to $z$ by the path reduction rule afterwards, then $yz\in B$ by the definition of the path reduction rule. 

Now suppose that $xz$ was added to $B$ due to an application of the path reduction rule. 
By definition of this rule, $x$ and $z$ were not adjacent to each other before.
Suppose that $yz\notin B$. 
Then $xy,yz$ are edges of the original input graph of the {\sc Tree $+\;k$ Edges Square Root} problem.
Because $xz$ was not such an edge,
$x$ and $y$ only became true twins due to an application of the path reduction rule. 
Then, by Lemma~\ref{l-twins}, $x$ or $y$ must be an outer vertex of some $F$-triple, that is, at least one of these two vertices must be incident to an edge of $R$. Then there is an edge of $R$ incident to at least one of these two vertices after the exhaustive application of the path reduction rule.
Because $x$ and $y$ are in $S_i\setminus X_i$, this is a contradiction to Lemma~\ref{l-claim1}.
Hence, $yz\in B$. This completes the proof of Lemma~\ref{l-claim2}.
\end{proof}

We prove the following lemma, which is our final lemma; in particular note that if $\hat{G}$ exists then its size is bounded by a quadratic function of~$k$.

\begin{lemma}\label{lem:final}
If the simplicial vertex reduction rule returned a no-answer for $(G,k,R,B)$, then $(G,k,R,B)$ has no solution.
Otherwise, the new instance $(\hat{G},k,\hat{R},\hat{B})$ has a solution if and only if $(G,k,R,B)$ has a solution. 
Moreover, $\hat{G}$ has at most $(15k-14)(15k-12)$ vertices.
\end{lemma}

\begin{proof}
We start by showing that $(G,k,R,B)$ has no solution if the simplicial vertex reduction rule returned a no-answer for $(G,k,R,B)$.
This can happen in step 2, 4, 5 or 6, each of which we discuss in a separate case.

\medskip
\noindent
{\bf Case 1.} The no-answer is given in step 2.  Suppose $(G,k,R,B)$ has a solution $H$. We will prove that $|V_G\setminus S|\leq 15k-14$, which means that returning a no-answer is correct if $|V_G\setminus S|>15k-14$.

Let $Z$ be the set of pendant vertices of $H$, and let $H^*=H-Z$. 
By Observation~\ref{obs:leaves} i), vertices in $Z$ are simplicial vertices of $G$. Then, by Lemma~\ref{lem:reduced}~ii) and~iii), we find that $Z\subseteq S$.  Hence, $|V_G\setminus S|=|V_G|-|S|=|V_H|- |S|\leq |V_H|-|Z|=|V_{H^*}|\leq 15k-14$, where the last inequality follows from Lemma~\ref{l-hstarbounded}.

\medskip
\noindent
{\bf Case 2.} The no-answer is given in step 4.
Suppose $(G,k,R,B)$ has a solution $H$. We will prove that $t \leq 15k-14$, which means that returning a no-answer is correct if $t>15k-14$.

Let $H^*$ be the graph obtained from $H$ after removing all pendant vertices of $H$. Then $|V_{H^*}|\leq 15k-14$ by Lemma~\ref{l-hstarbounded}.
If a set $S_i$ contains a pendant vertex $u$ of $H$, then $u$ is adjacent to a vertex $v$ of $H^*$.
Then, by Observation~\ref{obs:leaves} ii), $v$ is not adjacent to pendant vertices of $H$ in any $S_j$ with $j\neq i$.
Otherwise $S_i$ consists of non-pendant vertices of $H$, that is, vertices of $H^*$; being nonempty $S_i$ contains at least one vertex of $H^*$.
We conclude that every set in the true twin partition of $S$ corresponds to at least one unique vertex of $H^*$.
If their total number $t>15k-14$, this means that $|V_{H^*}|>15k-14$;  a contradiction. Hence,  $t\leq 15k-14$, as we had to show.

\medskip
\noindent
{\bf Case 3.} The no-answer is given in step 5. Suppose that $(G,k,R,B)$ has a solution $H$. We will prove that the edges of $R^1$ incident to 
a set $X_i$ have a common end-vertex for $i=1,\ldots,t$, which means that returning a no-answer is correct should this not be the case.

In order to obtain a contradiction, suppose that some set $X_i$ contains two vertices $u$ and $v$ 
that are incident to edges $uu', vv'\in R^1$ with $u'\neq v'$. 
 By Lemma~\ref{lem:reduced} ii), we find that $uu'$ and $vv'$ are incident to pendant vertices of $H$. 
 By Observation~\ref{obs:leaves}~iii), these pendant vertices are not adjacent in~$G$. 
However, from the definition of $S_i$ we deduce that $u,v,u',v'$ are mutually adjacent; a contradiction.
This completes Case~3.

\medskip
\noindent
{\bf Case 4.} The no-answer is given in step 6.  
Then there exists  a set $S_i$ such that $|S_i\setminus X_i|\geq 15k-13$ and such that there are 
three vertices $u\in X_i$, $v\in N_G(u)$ and $x\in S_i\setminus X_i$ with 
$uv\in R^1$ and $xv\in B$. In order to obtain a contradiction, assume that $(G,k,R,B)$ has a solution $H$. 

By Lemma~\ref{l-hstarbounded}, $H$ has at most $15k-14$ non-pendant vertices.  Because $|S_i\setminus X_i|\geq 15k-13$, this means that at least one vertex $y\in S_i\setminus X_i$ is a pendant vertex of $H$. 
Also, $u\in X_i$ is a pendant vertex of $H$ that has $v$ as its unique neighbor, because $uv\in R^1$ and all
other edges incident to $u$ belong to $B$ by definition of $S$.
If $y=x$, then $v$ is not adjacent to $y$ in $H$, because $xv\in B$.
If $y\neq x$, then $v$ is not adjacent to $y$ in~$H$ either, because $xv\in B$ and $B(x)=B(y)$ (due to Lemma~\ref{l-claim2}) imply $yv\in B$. We conclude that $u$ and $y$ are pendant vertices of $H$ adjacent to different vertices. However, from Observation~\ref{obs:leaves}~iii) we derive that $u$ and $y$ are not adjacent in $G$. This is a contradiction, because $u$ and $y$ are true twins in $G$ by definition of $S_i$. This completes Case~4.

\medskip
\noindent
From now on assume that the simplicial vertex reduction rule did not return a no-answer after performing 
step~6.
Let $(G',k,R',B')$ be the instance created after applying step~7 to some set $X_i=\{x_1,\ldots,x_{\ell}\}$ with $\ell\geq 2$, that is, $G'$ is the graph
obtained from $G$ after deleting $x_2,\ldots,x_\ell$, whereas $R'$ and 
$B'$ are the sets obtained from $R$ and $B$, respectively, after deleting edges incident to $x_2,\ldots,x_\ell$ from them.
We claim that $(G',k,R',B')$ has a solution if and only if $(G,k,R,B)$ has a solution. Before we prove this claim, we first observe that
in any solution $H$ for $(G,k,R,B)$ the vertices $x_1,\ldots,x_\ell$ are pendant vertices in $H$. This is because $x_1,\ldots,x_\ell$ are incident to exactly one edge in $R^1$, whereas all the other edges incident to them belong to $B$. Moreover, $x_1,\ldots,x_\ell$ have a (unique) common neighbor in $H$, as otherwise a no-answer would have been returned in step 5. We let $v$ denote this common neighbor. Similarly, $x_1$ is a pendant vertex that has 
$v$ as its (unique) neighbor in any solution $H'$ for $(G',k,R',B')$. 

First suppose that $(G',k,R',B')$ has a solution $H'$. 
Then the graph obtained from $H'$ by adding the vertices $x_2,\ldots,x_{\ell}$ and the edges $x_2v,\ldots,x_{\ell}v$ is a square root of $G$
by Lemma~\ref{lem:twins}~i). By definition of $R'$, $B'$ and the set $X_i$ (all of whose vertices are incident to one edge of $R^1\subseteq R$ and 
to edges in $B$) it is a  solution for $(G,k,R,B)$ as well.

Now suppose that $(G,k,R,B)$ has a solution $H$. 
Then the graph obtained from $H$ after deleting $x_2,\ldots,x_{\ell}$ is a square root of $G'$ by Lemma~\ref{lem:twins}~ii). By definition of $R'$ and $B'$, 
it is a  solution for $(G',k,R',B')$ as well.

\medskip
\noindent
We denote the instance resulting from 
step~7
by $(G,k,R,B)$ again and observe that every $X_i$ now contains at most one vertex. 
It remains to consider what happens at step~8. 
We let  $(G',k,R',B')$ be the instance created after applying step~8 to some set $S_i$ with $|S_i|>15k-13$, that is, 
$G'$ is the graph obtained from $G$ after deleting a set $T$ of $|S_i|-15k+13\geq 1$ arbitrary vertices from $S_i\setminus X_i$ (note that this is possible as $|X_i|\leq 1$),
whereas $R'$ and  $B'$ are the sets obtained from $R$ and $B$, respectively, after deleting the edges that are incident to
vertices of $T$. We claim that $(G',k,R',B')$ has a solution if and only if $(G,k,R,B)$ has a solution.

First suppose that $(G',k,R',B')$ has a solution $H'$.
Because we could not apply the trimming and path reduction rules for $(G,k,R,B)$, we cannot apply these rules for $(G',k,R',B')$ either.
Then, by using the same arguments that we applied for $(G,k,R,B)$ in the proof of Lemma~\ref{l-hstarbounded}, we find that $H'$ contains at most $15k-14$ non-pendant vertices.
Note that $H'$ contains at least $15k-13$ vertices, which are all in $S_i$.
Hence, $H'$ has at least one pendant vertex~$x$ that belongs to $S_i$. Let $v$ be the (unique) vertex adjacent to $x$ in $H'$. Then the graph $H$ obtained from $H'$ by adding the vertices of $T$ 
and their edges incident to $v$ 
is a square root of $G$ by Lemma~\ref{lem:twins}~i). We argue that $H$ is a solution for
$(G,k,R,B)$ as well. Because the vertices of $T\subseteq S_i\setminus X_i$ are not incident to the edges of $R$
 due to Lemma~\ref{l-claim1}, we have to show that none of the $|T|$ edges that we added in order to obtain $H$ belong to $B$. If $x\in S_i\setminus X_i$, then $xv\notin B$ and because $B(x)=B(y)$ for all $y\in S_i\setminus X_i$, we have that $yv\notin B$ for all $y\in T$. 
Assume that $x\in X_i$. Recall that $|X_i|\leq 1$ after step~7. Because $|S_i|>15k-13$ after step~7, $|S_i\setminus X_i|\geq 15k-13$. Then  $yv\notin B$ for all $y\in S_i\setminus X_i$ as otherwise the algorithm would have produced a no-answer at step~6. 

Now suppose that $(G,k,R,B)$ has a solution $H$.
By Lemma~\ref{l-hstarbounded}, the graph $H$ contains at most $15k-14$ non-pendant vertices. 
Hence, $H$ has at least $|S_i|-15k+14\geq 15k-12-15k+14=2$ pendant vertices.
Because vertices in $S_i\setminus X_i$ are true twins not incident to edges of $R$ and $B(x)=B(y)$ for any $x,y\in S_i\setminus X_i$, 
we may assume without loss of generality that the vertices of $T$ are amongst these pendant vertices of $H$.
If $X_i=\{x\}\neq\emptyset$, then $x$ is a pendant vertex in $H$ incident to a unique edge $xv\in R^1$. By Observation~\ref{obs:leaves}, 
 all pendant vertices of $H$ that are in $S_i$ are adjacent to $v$ in $H$. 
Then the graph obtained from $H$ after deleting the vertices of $T$ is a square root of $G'$ by Lemma~\ref{lem:twins}~ii). By definition of $R'$ and $B'$, it is a  solution for $(G',k,R',B')$ as well.
 If $X_i=\emptyset$, then all pendant vertices of $H$ that are in $S_i$ are adjacent to some $v$ in $H$ by Observation~\ref{obs:leaves}. 
Then, by Lemma~\ref{lem:twins}~(ii), the graph obtained from $H$ by deleting the vertices of $T$ is a square root of $G'$.
By definition of $R'$ and $B'$, it is a  solution for $(G',k,R',B')$ as well.

\medskip
\noindent
From the above it follows that the instance $(\hat{G},k,\hat{R},\hat{B})$ obtained after 
step~8 
has a solution if and only if $(G,k,R,B)$ has a solution.
In order to complete the proof, we must show that $\hat{G}$ has at most $(15k-14)(15k-12)$ vertices.
Each $S_i$ has at most $15k-13$ vertices due to 
step~8,
and we also have $t\leq 15k-14$ due to step 4. Hence $|S|\leq  (15k-14)(15k-13)$. As the number of vertices in $V_G\setminus S$ is at most $15k-14$ due to step 2, we obtain that $|V_{\hat{G}}|\leq (15k-14)(15k-13)+15k-14=
(15k-14)(15k-12)$, as required. 
\end{proof}

\subsection{Solving the Labeled Variant and Running Time Analysis}\label{s-solving}

Let $n$ and $m$ denote the number of vertices and edges of the graph $G$ of the original instance $(G,k)$ of 
{\sc Tree $+\;k$ Edges Square Root}. 
In order to complete the proof of Theorem~\ref{thm:tree-few-edges},
we first note that the trimming and path reduction rules are applied at most $n$ times to construct the instance $(\hat{G},k,\hat{R},\hat{B})$. Each application of the trimming rule can be done in time $O(n^2m)$ and each application of the path reduction rule takes time $O(n^3m)$. Finally, the simplicial vertex reduction rule can be done in time $O(nm)$. 
Hence, our kernelization algorithm runs in time $O(n^4m)$, and it
remains to solve the obtained reduced instance 
$(\hat{G},k,\hat{R},\hat{B})$.
Because $\hat{G}$ has at most  $(15k-14)(15k-12)$ vertices, $\hat{G}$ has at most  $\frac{1}{2}(15k-14)(15k-12)( (15k-14)(15k-12)-1)=O(k^4)$ edges. Therefore, we can solve {\sc Tree $+\;k$ Edges Square Root with Labels} for instance $(\hat{G},k,\hat{R},\hat{B})$
 in time $2^{O(k^4)}$; we consider all edge subsets of $\hat{G}$ that have size at most $|V_{\hat{G}}|-1+k$  and use brute force.  
We conclude that the total running time of our algorithm is $2^{O(k^4)}+O(n^4m)$, as required.

\medskip
\noindent
We finish this section with the following remarks. First, recall that our quadratic kernel  is a generalized kernel for the {\sc Tree $+\;k$ Edges Square Root} problem. We believe that a quadratic kernel exists for this problem as well by using a similar reduction. However, proving this seemed to be more technical and also to yield a graph with more than $(15k-14)(15k-12)$ vertices. We therefore chose to prove our \FPT\ result by using a reduction leading to a generalized kernel. Second, it should also be noted that our generalized kernel for {\sc Tree $+\;k$ Edges Square Root} does not imply a kernel for {\sc Tree $+\;k$ Edges Square Root with Labels}, 
because our reduction rules require that the original instance is unlabeled.  
We do not know whether the (more general) problem {\sc Tree $+\;k$ Edges Square Root with Labels} is \FPT\ as well.

\section{The Maximum Square Root Problem}\label{s-max}

Recall that the {\sc Maximum Square Root} problem is that of testing whether a given graph $G$ with $m$ edges has a square root
with at least $s$ edges for some given integer $s$.
In this section we give an \FPT\ algorithm for this problem with parameter $k=m-s$. In other words, we show that the problem of 
deciding whether a graph $G$ has a square root that can be obtained by removing at most $k$ edges of $G$ is fixed-parameter tractable when
parameterized by $k$. 
We also present an exact algorithm for
the {\sc Maximum Square Root} problem.
Both algorithms are
based on the observation that in order to construct a square root $H$ from a given graph $G$,  we must delete at least one of every pair of adjacent edges
that do not belong to a triangle in $G$.  
We therefore construct an auxiliary graph
${\cal P}(G)$ that has vertex set $E_G$ and an edge between two vertices $e_1$ and $e_2$ if and only if
$e_1=xy$ and $e_2=yz$ for three distinct vertices $x,y,z\in V_G$ with $xz\notin E_G$.
Observe that ${\cal P}(G)$ is a spanning subgraph of the line graph of $G$.
We need the following lemma.

\begin{lemma} \label{lem:charact}
Let  $H$ be a spanning subgraph of a graph $G$. 
Then $H$ is a square root of $G$ if and only if $E_H$ is an independent set of ${\cal P}(G)$
and every two adjacent vertices in $G$  are at distance at most $2$ in $H$.
\end{lemma}

\begin{proof}
First suppose that $H$ is a square root of $G$. By definition,  every two adjacent vertices in $G$ are of distance at most 2 in $H$.
In order to show that $E_H$ is an independent set in ${\cal P}(G)$, assume that two edges $e_1,e_2\in E_H$ are adjacent vertices in 
 ${\cal P}(G)$. Then $e_1=xy$ and $e_2=yz$ for three distinct vertices $x,y,z\in V_G$ with $xz\notin E_G$. This means that $x$ and $z$ are of distance~2 in $H$ implying that $xz\in E_G$, which is a contradiction.

Now suppose that $E_H$ is an independent set of ${\cal P}(G)$ 
and that every two adjacent vertices in $G$  are at distance at most $2$ in $H$.
In order to show that $H$ is a square root of $G$, it suffices to show that 
every two non-adjacent vertices in $G$ have  distance at least 3 in $H$.
Let $u$ and $v$ be two non-adjacent vertices in $G$ that have distance at most 2 in $H$. 
Then there exists a vertex $z\notin \{u,v\}$ such that $uz,vz\in E_H$.
Then $e_1=uz$ and $e_2=vz$ are adjacent in ${\cal P}(G)$ contradicting the independence of $E_H$ in ${\cal P}(G)$.
\end{proof}

We use Lemma~\ref{lem:charact} to prove Propositions~\ref{p-fpt} and~\ref{p-exact}. Here, we use the $O^*$-notation to suppress any polynomial factors. A {\it vertex cover} is a subset $U\subseteq V$ such that every edge is incident with at least one vertex of $U$.
The {\sc Vertex Cover} problem is that of testing whether a given graph has a vertex cover of size at most $p$ for a given integer $p$.

In Proposition~\ref{p-fpt} we prove that there is a $O^*(2^k)$ time algorithm
to decide whether a given graph $G$ has square root $H$ such that $|E_G\setminus E_H|\le k$.

\begin{proposition}\label{p-fpt}
{\sc Maximum Square Root} can be solved in time $O^*(2^k)$.
\end{proposition}

\begin{proof}
Let $G$ be a graph with $n$ vertices and $m$ edges, and let $k\geq 0$ be an integer.
By Lemma~\ref{lem:charact} it suffices to check whether ${\cal P}(G)$ has a vertex cover $U$
of size at most $k$ such that $H_U=(V_G,E_G\setminus U)$ is a square root of $G$. All vertex covers of size at most $k$
of a graph can be enumerated by adapting the standard $O^*(2^k)$ branching algorithm
for the {\sc Vertex Cover} problem (see for example~\cite{DowneyF99}). It requires $O(m^2)$ time to compute ${\cal P}(G)$ and $O(nm)$ time to
check whether a graph $H_U$ is a square root of $G$. Hence the overall running time of our algorithm is  $O^*(2^k)$.
\end{proof}

We observe that {\sc Maximum Square Root} has a linear kernel 
for connected graphs. 
This immediately follows from a result 
of Aingworth, Motwani and Harary~\cite{AingworthMH98}, who proved that if $H$ is a square root of a connected  
$n$-vertex graph $G\neq K_n$, then $|E_G\setminus E_H|\geq n-2$.  Hence, $n\leq k+2$ for every yes-instance $(G,k)$ of {\sc Maximum Square Root} with $G\neq K_n$ (trivially,  $K_n$ is its own square root). Note that this kernel does not lead to a faster running time than $O^*(2^k)$.

In Proposition~\ref{p-exact} we present our exact algorithm, which does not only solve the decision problem but in fact determines a square root of a given graph that has maximum number of edges. 

\begin{proposition}\label{p-exact}
{\sc Maximum Square Root} can be solved in time $O^*(3^{m/3})$ on graphs with $m$ vertices. 
\end{proposition}

\begin{proof}
Let $G$ be a graph with $n$ vertices and $m$ edges, and let $k\geq 0$ be an integer.
We compute the graph ${\cal P}(G)$,
enumerate all maximal independent sets $I$ of  ${\cal P}(G)$, and verify for each 
$I\subseteq E_G$ whether $G$ is the square of the graph $H_I=(V_G,I)$.
Out of those graphs  $H_I$ that are square roots
of $G$, return the one with maximum number edges; if no such graph $H_I$ has been found, then $G$ has no square roots.  
Correctness follows from Lemma~\ref{lem:charact}. 
Recall that ${\cal P}(G)$ can be computed in time $O(m^2)$.
All the maximal independent sets of the $m$-vertex graph ${\cal P}(G)$ can be enumerated
in time $O^*(3^{m/3})$ using the polynomial delay algorithm of 
Tsukiyama et al.~\cite{TsukiyamaIAS77}, since ${\cal P}(G)$ has at most 
$3^{m/3}$ maximal independent sets~\cite{MoonM65}. Finally, recall that
for each maximal independent set $I$, we can check in time $O(nm)$ 
whether $(H_I)^2=G$.
Hence the overall running time of our algorithm is $O^*(3^{m/3})$.
\end{proof}

\section{Open Problems}\label{s-con}

We conclude our paper with two open problems.
First, is it also possible to construct an exact algorithm for {\sc Minimum Square Root} that is better than the trivial exact algorithm?

Second, recall that if $H$ is a square root of a connected  
$n$-vertex graph $G\neq K_n$, then $|E_G\setminus E_H|\geq n-2$~\cite{AingworthMH98}. 
Is it \FPT\ to decide whether a connected $n$-vertex graph $G\neq K_n$ 
has a square root that can be obtained by removing at most $n-2+k$ edges, or equivalently, whether a connected $n$-vertex graph $G\neq K_n$
has a square root with at least $|E_G|-n+2-k$ edges, when parameterized by $k$? 
In particular, can it be decided in polynomial time whether a connected graph $G$ has a square root with \emph{exactly} $|E_G|-|V_G|+2$ edges? 


\end{document}